\documentclass[a4paper,conference]{IEEEtran}

\usepackage{amsmath}
\usepackage{amssymb}
\usepackage{graphicx}
\usepackage[latin1]{inputenc}
\usepackage{color}

\newcommand{\argmin}[1]{\underset{#1}{\operatorname{argmin}}}

\newcommand{\E}[1]{\mathrm{E}[ #1 ]}
\newcommand{\cE}[2]{\mathrm{E}[ #1 \mid #2 ]}
\newcommand{\PP}[1]{\mathrm{P}\{ #1 \}}
\newcommand{\bz}{\mathbf{z}}
\newcommand{\bZ}{\tilde{\mathbf{z}}}

\newcommand{\cost}{\omega}
\newcommand{\poli}[1]{#1}
\newcommand{\dps}{\displaystyle}

\newcommand{\hc}{b}  %
\newcommand{\HC}{B}  %
\newcommand{\ST}{T}  %
\newcommand{\JS}{X}  %

\newtheorem{theorem}{Theorem} %
\newtheorem{lemma}{Lemma}
\newtheorem{prop}{Proposition}
\newtheorem{corollary}{Corollary}

\newtheorem{defn}{Definition}

\begin{document}

\IEEEoverridecommandlockouts

\title{Minimizing Slowdown in Heterogeneous\\
  Size-Aware Dispatching Systems (full version)$^{\dagger}$\thanks{$^{\dagger}$ 
    This is the full version (including the appendix) of the paper with the same title that appears in the ACM SIGMETRICS 2012, London, UK.}}

\author{
  \IEEEauthorblockN{Esa Hyyti{\"a}, Samuli Aalto and Aleksi Penttinen}
  \IEEEauthorblockA{Department of Communications and Networking\\Aalto University School of Electrical Engineering, Finland}}
\maketitle

\begin{abstract}
We consider a 
system of parallel queues where
tasks are assigned (dispatched) to one of the available servers upon arrival.
The dispatching decision is based on the full state information, i.e.,
on the sizes of the new and existing jobs.
We are interested in minimizing the so-called mean slowdown criterion
corresponding to the mean of the sojourn time divided by the processing time.
Assuming no new jobs arrive, 
the shortest-processing-time-product (SPTP) schedule is known to minimize the
slowdown of the existing jobs.
The main contribution of this paper is three-fold:
1) To show the optimality of SPTP with respect to slowdown 
in a single server queue \emph{under Poisson arrivals}; 
2) to derive the so-called size-aware value functions for
M/G/1-FIFO/LIFO/SPTP/SPT/SRPT
\emph{with general holding costs}
of which the slowdown criterion is a special case; and
3) to utilize the value functions to derive efficient dispatching policies
so as to minimize the mean slowdown in a heterogeneous server system.
The derived %
policies %
offer a significantly better
performance than 
e.g., the size-aware-task-assignment with equal load (SITA-E) and least-work-left (LWL) policies.
\end{abstract}

\section{Introduction}

Dispatching problems arise in many contexts such as manufacturing sites,
web server farms, super computing systems, and other parallel server systems.
In a dispatching system jobs are assigned upon arrival to one of the several queues
as illustrated in Fig.~\ref{fig:dispatch-m}.
Such systems involve two decisions:
(i) \emph{dispatching policy} $\alpha$ chooses the server, and
(ii) \emph{scheduling discipline} the order in which the jobs are served. %
The dispatching decisions are irrevocable, 
i.e., it is not possible to move a job to another queue afterwards.
In the literature, the dispatching problems and their solutions 
differ with respect to 
  (i) optimization objective,
 (ii) available information, and %
(iii) scheduling discipline used in the servers.

Although the literature has generally addressed the mean sojourn time (i.e., response time) %
as the optimization objective,  also other performance metrics can be relevant. %
One such metric is the \emph{slowdown}\footnote{Yang and de Veciana refer to
the slowdown as the bit-transmission delay (BTD)
\cite{yang-infocom-2002}.} %
of a job, defined as the ratio of the sojourn time
and the processing time (service requirement) \cite{harchol-balter-peva-2002,feng-sigmetrics-2003,aalto-sigmetrics-2007}.
The slowdown criterion combines
\emph{efficiency} and \emph{fairness} 
and stems from the idea that
\emph{longer jobs can tolerate longer sojourn time}.
The optimal scheduling discipline in this respect for a single server queue
and a fixed number of jobs is %
the so-called \emph{shortest-processing-time-product}%
\footnote{Wierman et al.\ refer to SPTP as the RS policy,
  a product of the remaining size and the original size \cite{wierman-sigmetrics-2005}.}
(SPTP) discipline \cite{yang-infocom-2002},
where the index of a job is the product of the initial and remaining service requirements
and the job with the smallest index is served first.
With aid of Gittins index,
we show that SPTP minimizes
the mean slowdown in single server queues also in the dynamic case of Poisson arrivals,
i.e., 
it is the optimal discipline for an M/G/1 queue with respect to %
the %
slowdown.

\begin{figure}
\centering
\includegraphics[height=23mm]{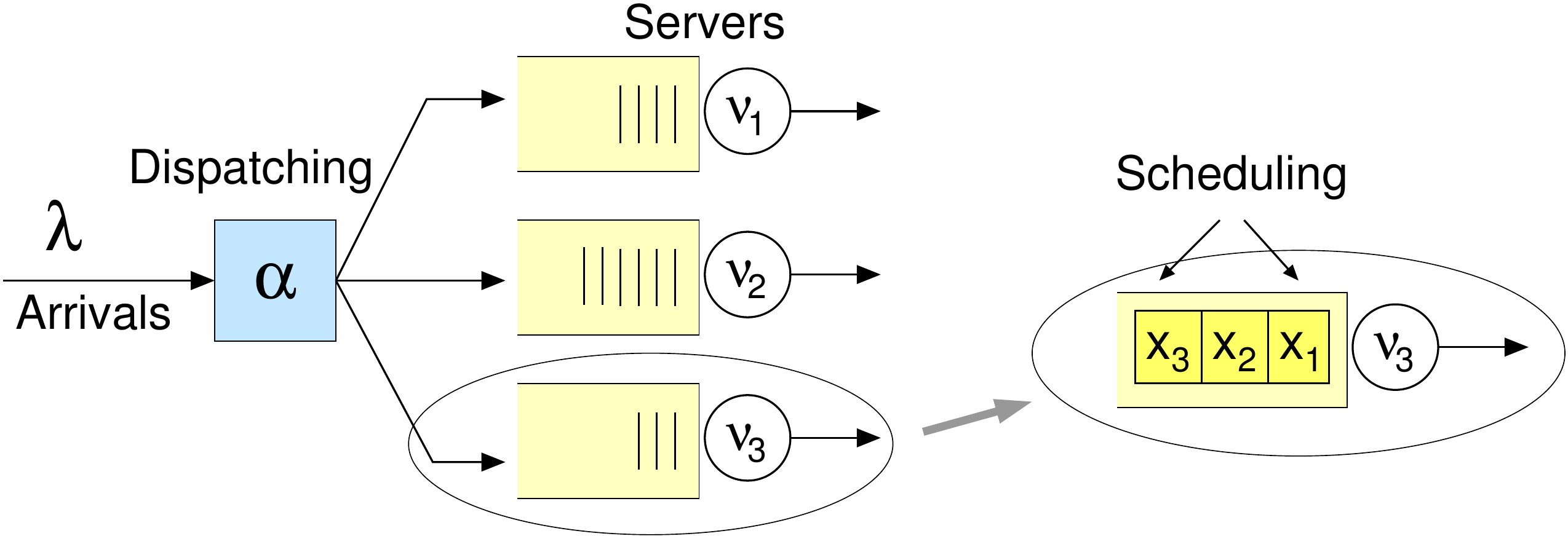}
\caption{A dispatching system with $m=3$ servers.}
\label{fig:dispatch-m}
\end{figure}

Then we %
derive value functions
with respect to \emph{arbitrary job specific holding costs} for an M/G/1 with 
the first-in-first-out (FIFO),
the last-in-first-out (LIFO), 
the shortest-processing-time (SPT), the shortest-remaining-processing-time (SRPT)
and SPTP scheduling disciplines,
where the last three are size-aware. %
A value function essentially characterizes the %
queue state in terms of
the expected %
costs in infinite time horizon for a fixed scheduling discipline.
In this respect, our work %
generalizes the results of \cite{hyytia-ejor-2012} to arbitrary
job specific holding cost rates, 
and includes also the analysis of the slowdown specific SPTP 
scheduling discipline.

Finally, we apply the derived value functions %
to the dispatching problem so as to minimize the mean slowdown
when the scheduling disciplines in each queue are fixed.
We assume that the dispatcher is fully aware of the state of the system, i.e.,
of the tasks in each queue and their remaining service requirements.
By starting with an arbitrary state-independent policy,
we carry out the first-policy-iteration (FPI) step of the Markov
decision processes (MDP) %
and obtain efficient %
state-dependent dispatching policies.

The rest of the paper is organized as follows.
In Section~\ref{sec:mg1-slow}, we consider a single M/G/1 queue with respect to slowdown criterion, %
and prove the optimality of SPTP. %
In Section~\ref{sec:values}, we derive the size-aware value functions with respect to
arbitrary job specific holding cost rates for FIFO, LIFO %
and SPTP (SPT and SRPT are given in the Appendix).
The single queue scheduling related results are utilized in Section~\ref{sec:dispatching}
to derive efficient dispatching policies, which are then
evaluated in Section \ref{sec:examples}.
Section~\ref{sec:conclusions} concludes the paper.

\subsection{Related Work}
\label{sect:related}
FIFO is perhaps the most common scheduling %
discipline
due to its nature and ease of implementation.
Other common %
disciplines are LIFO, SPT and SRPT.
As for the objective, minimization of the mean sojourn time has been a popular choice.
Indeed, SPT and SRPT, respectively,
are the optimal non-preemptive and preemptive schedules with this respect \cite{schrage-or-1968}.
For a recent survey on fairness and other scheduling objectives 
in a single server queue we refer to \cite{wierman-or-survey-2011}.

In the context of dispatching problems, 
FIFO has been studied extensively in the literature since
the early work by Winston \cite{winston-applied-1977},
Ephremides et al.\ \cite{ephremides-tac-1980}, and others.
Often the number of tasks per server is assumed to be known, cf., e.g., 
\emph{join-the-shortest-queue} (JSQ) dispatching policy \cite{winston-applied-1977}.
Even though FIFO queues have received the most of the attention,
also other scheduling disciplines have been studied.
For example, 
Gupta et.\ al consider JSQ with 
\emph{processor-sharing} (PS) scheduling discipline in \cite{gupta-peva-2007}.

Only a few optimality results are known for the dispatching problems.
Assuming exponentially distributed interarrival times and job sizes,
\cite{winston-applied-1977}
shows that JSQ with FIFO minimizes the mean waiting time
when the number in each queue is available. %
Also \cite{ephremides-tac-1980}
argues for the optimality of JSQ/FIFO when the number in each queue is available, 
while the Round-Robin (RR), followed by FIFO,
is shown to be the optimal policy when 
it is only known
that the queues were initially in the same state.
\cite{liu-applied-1994} proves that RR/FIFO is optimal 
with the absence of queue length information 
if the job sizes have a non-decreasing hazard function.
The RR results were later generalized in \cite{liu-or-1998}.
Whitt \cite{whitt-or-1986}, on the other hand, provides several
counterexamples where JSQ/FIFO policy fails.
Crovella et al.\cite{crovella-sigmetrics-1998}
and Harchol-Balter et al.\ \cite{harchol-balter-pdc-1999}
assume that the dispatcher is aware
of the size of a new job, but not of the state of the FIFO queues,
and propose policies based 
on job size intervals (e.g., short jobs to one queue, and the rest to another).
Feng et al.\ \cite{feng-peva-2005} later
showed that such a policy is the optimal
size-aware state-independent dispatching policy for homogeneous servers.

Also the MDP framework lends itself to dispatching problems.
Krishnan \cite{krishnan-ieee-ac-1990} has utilized it 
in the context of parallel M/M/s-FIFO servers 
so as to minimize the mean sojourn time, similarly as
Aalto and Virtamo \cite{aalto-nts13-1996} for the traditional M/M/1 queue.
Recently, %
FIFO, LIFO, SPT and SRPT queues %
were analyzed in \cite{hyytia-ejor-2012} with a general service time distribution.
Similarly, PS is considered in \cite{hyytia-peva-2011,hyytia-itc-2011}.
The key idea with the above %
work is to start by an arbitrary state-independent policy,
and then carry out FPI step utilizing 
the value functions (relative values of states).

\section{M/G/1 Queue and Slowdown}
\label{sec:mg1-slow}

In this section we
consider a single M/G/1 queue with arrival rate $\lambda$.
The service requirements
are i.i.d.\ random variables $\JS_i\sim \JS$ with a general distribution.

We define the \emph{slowdown} of a job as a ratio of the sojourn time $\ST$
to the service requirement $\JS$,
$$
\gamma \triangleq \frac{\ST}{\JS},
$$
where one is generally interested in its mean value $\E{\gamma}$.
Also other similar definitions exist.
In the context of distributed computing,
Harchol-Balter defines the slowdown as the ``wall-time'' divided by the CPU-time
in \cite{harchol-balter-cs-1997}.
Another common convention is to consider the ratio of the waiting time
to the processing time (see, e.g., \cite{harchol-balter-jacm-2002}),
which is a well-defined quantity for non-preemptive systems.
These different definitions, however, are essentially the same.

For non-preemptive work conserving policies, the
sojourn time in queue comprises the initial waiting time $W$ and
the consequent processing time $\JS$.
Thus, the mean slowdown is
$$
\E{\gamma} = \E{\frac{W+\JS}{\JS}} = 1 + \E{W/\JS}.
$$
Waiting time with FIFO is independent of the job size $\JS$,
and with aid of 
Pollaczek-Khinchin formula
one obtains \cite{harchol-balter-jacm-2002},
\begin{equation}\label{eq:fifo-slowdown}
\E{\gamma} 
= 1 + \E{W}\cdot\E{\JS^{-1}}
= 1 + \frac{\lambda\,\E{\JS^2}}{2(1-\rho)}\cdot\E{\JS^{-1}},
\end{equation}
which %
underlines the fact that the mean slowdown 
may be infinite for a stable queue as
$\E{\JS^{-1}}=\infty$ for many common job size distributions,
e.g., for an exponential distribution. %

The mean slowdown in an M/G/1 queue with preemptive LIFO and PS disciplines 
is a constant for all job sizes,
\begin{equation}\label{eq:lifo-slowdown}
 \cE{\gamma}{\JS=x} = \frac{1}{1-\rho}.
\end{equation}
Comparison of \eqref{eq:fifo-slowdown} and \eqref{eq:lifo-slowdown} immediately gives:
\begin{corollary}
\label{cor:fifo-lifo-slowdown}
FIFO is a better scheduling discipline than LIFO in an M/G/1 queue 
with respect to slowdown iff
\begin{equation}\label{eq:fifo-vs-lifo}
2\,\E{\JS} > \E{\JS^2} \cdot \E{\JS^{-1}}.
\end{equation}
\end{corollary}

\subsection{Shortest-Processing-Time-Product (SPTP)}

In \cite{yang-infocom-2002}, Yang and de Veciana introduced the 
\emph{shortest-processing-time-product} (SPTP)
scheduling discipline,
which will serve the job $i^*$ such that 
\[
  i^* = \arg\min_i \Delta^*_i \Delta_i,
\]
where $\Delta_i^*$ and $\Delta_i$ denote, respectively, the initial and remaining
service requirement of job $i$.
Yang and de Veciana
were able to show that SPTP is optimal with respect to the mean slowdown 
$\E{\gamma}$ in a transient system where all jobs are available at time 0 and 
no new jobs arrive thereafter (i.e., a \emph{myopic} approach).
SPTP can be seen as the counterpart of the SRPT\footnote{%
Note that defining an index policy with indices $\Delta_i$, $\Delta_i^*$ and $\sqrt{\Delta_i\Delta_i^*}$
gives SRPT, SPT and SPTP, respectively, where $\sqrt{\Delta_i\Delta_i^*}$ corresponds to the geometric mean
of $\Delta_i$ and $\Delta_i^*$.}
when instead of minimizing the mean sojourn time,
one is interested in minimizing the mean slowdown.

We argue below that SPTP is the optimal scheduling discipline also in 
the corresponding dynamic system with Poisson arrivals. %
That is,
we show that SPTP is optimal with respect to the %
slowdown %
in the M/G/1 queue. 
Our proof is based on the Gittins index approach.

\subsection{Gittins Index}
\label{sec:gittins}

Consider an M/G/1 queue with {\em multiple} job classes $k$, $k = 1,\ldots,K$. 
Let $\lambda_k$ denote the class-$k$ arrival rate, and $p_k = \lambda_k/\lambda$ 
the probability that an arriving job belongs to class $k$. In addition, let $\JS_k$ 
denote the generic service time related to class $k$. The total load is denoted 
by $\rho = \lambda_1 \E{\JS_1} + \ldots + \lambda_K \E{\JS_K}$. We assume that $\rho < 1$. 
The {\em Gittins index} \cite{gittins-1989,aalto-queueing-2009} 
for a class-$k$ job with attained service $a$ 
is defined by 
\[
  G_k(a) = 
  \sup_{\delta > 0} 
  w_k \, \frac{\PP{\JS_k - a \le \delta \mid \JS_k > a}}{\cE{\min\{\JS_k - a, \delta \} }{\JS_k > a}}, 
\]
where $w_k$ is the holding cost rate related to class $k$. Note that, 
in addition to the attained service $a$, the Gittins index is based on the 
(class-specific) distribution of the service time, but not on the service time 
itself. The {\em Gittins index policy} will serve the job $i^*$ such that 
\[
  i^* = \arg\max_i G_{k_i}(a_i), 
\]
where $k_i$ refers to the class and $a_i$ to the attained service of job $i$.
Let $\ST_k$ denote the sojourn time of a class-$k$ job. We recall the following 
optimality result proved 
by Gittins \cite[Theorem~3.28]{gittins-1989}.

\begin{theorem}
\label{thm:gittins} 
Consider an $M/G/1$ multi-class queue. The Gittins index policy minimizes the 
mean holding costs, 
\[
  \sum_k p_k w_k \E{\ST_k},
\]
among the non-anticipating scheduling policies.
\end{theorem}

The non-anticipating policies are only aware of the attained service times 
$a_i$ of each job $i$ and the service time distributions (but not on the 
actual service times). So, in general, non-anticipating policies do not know 
the remaining service times implying that, e.g., SPTP does not belong to 
non-anticipating disciplines. However, there is one exception: If the service 
times $\JS_k$ are deterministic, $\PP{\JS_k = x_k} = 1$, then the remaining service 
times $x_k - a_k$ are available for the non-anticipating policies (with 
probability~1). In such a case, all policies are non-anticipating.

\subsection{Optimality of SPTP}
\label{sec:optimality}

Consider now an M/G/1 queue with a {\em single} class and load $\rho < 1$. Let 
$\JS$ and $\ST$ denote, respectively, the generic service and sojourn times of a job. 
In addition, let $\tau(x)$ denote the conditional mean sojourn time 
of a job with service time $x$,
$$
\tau(x) \triangleq \cE{ \ST }{ \JS=x }.
$$
Assume first that the support of the service time distribution is finite. 
In other words, there are $x_1 < \ldots < x_K$ such that 
$\sum_k \PP{\JS = x_k} = 1$. 
The following result is an immediate consequence of Theorem~\ref{thm:gittins} 
as soon as we associate class $k$ with the jobs with the same service time 
requirement $x_k$. Note that the Gittins index is now clearly given by 
\[
  G_k(a) = \frac{w_k}{x_k - a} 
\]
with the optimal $\delta$ equal to $x_k - a$.

\begin{corollary}
\label{cor:fin-sup} 
Consider an $M/G/1$ queue for which the support of the service time distribution 
is finite. Among all scheduling policies, the mean holding costs, 
\[
  \sum_k \PP{\JS = x_k} w_k \tau(x_k), 
\]
are minimized by the policy that will serve job $i^*$ such that 
\[
  i^* = \arg\min_i \frac{\Delta_i}{w_{k(i)}}
\]
where $k(i)$ and $\Delta_i$ denote the class and the remaining
service requirement of job $i$.
\end{corollary}

In particular, if we choose 
$w_k=1/x_k$,
we see that the mean slowdown is 
minimized by SPTP:

\begin{corollary}
\label{cor:sptp-opt} 
Consider an $M/G/1$ queue for which the support of the service time distribution 
is finite.
SPTP minimizes the mean slowdown $\E{\gamma}$ among 
all scheduling policies.
\end{corollary}

Since any service time distribution can be 
approximated with an arbitrary precision by discrete service time 
distributions, we %
expect that the above result holds also for general %
service time distributions.
Note also that the choice $w_k = 1$, for all $k$, results in the well-known 
optimality result of SRPT with respect to the mean sojourn time.

\section{Value Functions for M/G/1}
\label{sec:values}

In this section we analyze a single M/G/1 queue in isolation
with FIFO, LIFO, and SPTP scheduling disciplines
and derive expressions for the size-aware value functions 
\emph{with respect to arbitrary job specific holding costs}.
The corresponding results for SPT and SRPT are also given in the Appendix.
As special cases, one trivially obtains the value functions
with respect to the mean sojourn time and the mean slowdown.
Thus, these results generalize the corresponding results in \cite{hyytia-ejor-2012},
where the size-aware value functions
with respect to sojourn time are given (except for SPTP). 

\subsubsection*{Holding cost model}
Consider an M/G/1 queue with an arbitrary but fixed work conserving
scheduling discipline and load $\rho<1$.
Each existing task incurs costs at some \emph{task specific holding cost rate}
denoted by $\HC_i$ for job $i$.
Assume that $\HC_i$ are i.i.d.\ random variables, $\HC_i\sim\HC$.
However, $\HC_i$ may depend on the corresponding service requirement $\JS_i$.
Both the service requirement and the holding cost
become known upon arrival.

The cumulative cost during $(0,t)$ for an initial state $\bz$
is %
$$
V_{\bz}(t) \triangleq \int_0^t \sum_{i\in\mathcal{N}_{\bz}(s)} \HC_i \,ds,
$$
where $\mathcal{N}_{\bz}(s)$ denotes the set of tasks present in the system at time $s$.
Similarly, the long-term mean cost rate is
$$
r_{\bz}
\triangleq \lim_{t\to\infty} \frac{1}{t}\,\E{ V_{\bz}(t) } 
= \lim_{t\to\infty} \frac{1}{t}\,\E{ \int_0^t \sum_{i\in\mathcal{N}_{\bz}(s)} \HC_i \,ds},
$$
which for an ergodic Markovian system reads
$$
r = \lim_{t\to\infty} \E{ \sum_{i\in\mathcal{N}_{\bz}(t)} \HC_i}.
$$
The value function $v_{\bz}$ is defined as
the expected deviation from the mean cost rate
in infinite time-horizon,
\begin{align*}
v_{\bz} &\triangleq \lim_{t\to\infty} \E{V_{\bz}(t) - r\cdot t}.
\end{align*}
A value function characterizes how %
expensive it is to start from 
each state. %
In our setting,
it enables one to compute the expected cost of admitting
a job with size $x$ and holding cost rate $\hc$ to a queue,
which afterwards behaves as an M/G/1 queue
with the given scheduling discipline:
\begin{equation}\label{eq:cost}
\cost_{\bz}(x,\hc) \triangleq v_{\bz\oplus(x,\hc)} - v_{\bz},
\end{equation}
where $\bz\oplus(x,\hc)$ denotes the state resulting from adding a new job $(x,\hc)$ in state $\bz$.
In Section~\ref{sec:dispatching} we will carry out the FPI
step in the context of dispatching problem, 
for which a corresponding value function %
is a prerequisite.

\subsubsection*{Size dependent holding cost}
In general, the holding cost rate $b$ can be arbitrary, e.g., a class-specific
i.i.d.\ random variable. 
However, %
in two important special cases, 
mean sojourn time and mean slowdown,
it depends solely on the service requirement $x$, $b=c(x)$:
\begin{align*}
&\text{mean sojourn time: }    & c(x) &= 1\\
&\text{mean slowdown: }        & c(x) &= 1/x
\end{align*}
In equilibrium, the average cost \emph{per job} is
$\E{c(\JS)\cdot \ST}$.
Using Little's result, we have
for the \emph{mean cost rate},
$$
r
 = \lambda\,\E{c(\JS)\cdot \ST}
 = \left\{
\begin{array}{ll}
 \lambda\, \E{\ST} = \E{N}, & \text{for $c(x)=1$,}\\
 \lambda\, \E{\gamma},      & \text{for $c(x)=1/x$.}
\end{array}\right.
$$
where $\E{N}$ denotes the mean number in the system.

\subsection{M/G/1-FIFO Queue}

Next we derive the size-aware value function for an M/G/1-FIFO queue
with respect to arbitrary job specific holding costs.
To this end,
let $\bz=((\Delta_1,\hc_1); \ldots ;(\Delta_n,\hc_n))$ denote the remaining service requirements (measured in time)
$\Delta_i$ and the corresponding \emph{holding cost rates} $\hc_i$ at state $\bz$.
The total rate at which costs are accrued at given state $\bz$ is thus the sum $\sum_i \hc_i$.
Job $1$ is currently receiving service and job $n$ is the latest arrival.
The total backlog is denoted by $u_{\bz}$,
$$
u_{\bz} = \sum_{i=1}^n \Delta_i.
$$
\begin{prop}%
For the size-aware relative value in an M/G/1-FIFO queue with respect
to arbitrary job specific holding costs it holds that
\begin{equation}\label{eq:vn-fifo}
v_{\bz}-v_0 = \sum_{i=1}^{n} \left( \hc_i \sum_{j=1}^i \Delta_j\right)
 + \frac{\lambda\cdot\E{\HC}}{2\,(1-\rho)} u_{\bz}^2
\end{equation}
where $\E{\HC}$ is the mean holding cost rate for all later jobs.
\end{prop}
\begin{proof}
We compare two systems with the same arrival patterns, 
System 1 initially in state $\bz$ and System 2 initially empty.
The two systems behave equivalently after System 1 becomes empty. %
The cost incurred by the 
current $n$ jobs, present only in System 1, is already fixed as
$$
h_1=\sum_{i=1}^{n} \left( \hc_i \sum_{j=1}^i \Delta_j\right).
$$
The later 
arriving jobs encounter a longer waiting time in System 1.  
The key observation here is that these jobs experience an additional delay of $Y$ 
in System 1 when compared to System 2, as is illustrated in Fig.~\ref{fig:vn-fifo}.
Otherwise the %
sojourn times are equal.
On average $\lambda u_{\bz}$ (mini) busy periods occur 
before System 1 becomes empty.
These busy periods are independent and on average
$1/(1-\rho)$ jobs are served during each of them.
The average additional waiting time is $\E{Y}=u_{\bz}/2$.
Therefore, the later arriving jobs incur on average
$$
h_2 = \frac{\lambda\cdot\E{\HC}}{2\,(1-\rho)} u_{\bz}^2
$$
higher holding cost in System 1 than in System 2.
Total cost difference is $h_1+h_2$, which completes the proof.
\end{proof}

\begin{figure}
\centering
\includegraphics[width=60mm]{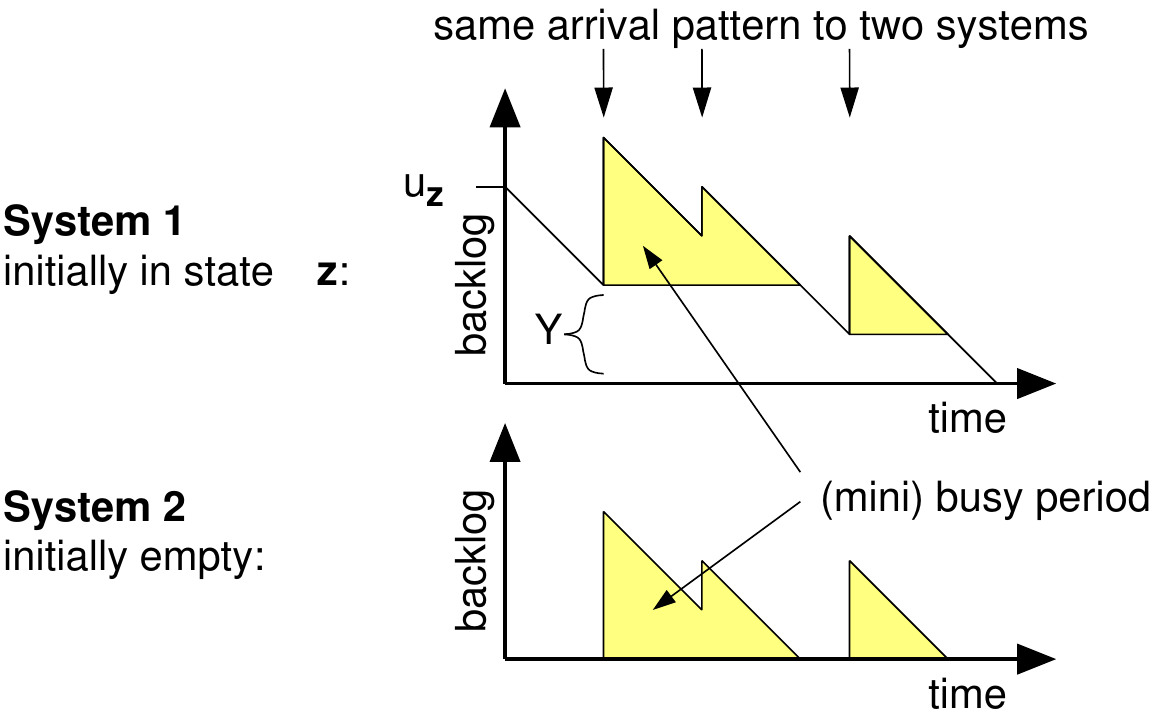}
\caption{Derivation of the value function in an M/G/1-FIFO queue.}
\label{fig:vn-fifo}
\end{figure}

\begin{corollary}%
The mean cost in terms of sojourn time due to accepting a job with size $x$ to 
a size-aware M/G/1-FIFO queue initially at state $\bz$ is
$$
\cost_\bz(x) = x + u_{\bz} + \frac{\lambda}{2(1-\rho)}(2\,u_{\bz}\,x+x^2).
$$
\end{corollary}

\begin{corollary}%
The mean cost in terms of slowdown due to accepting a job with size $x$ to 
a size-aware M/G/1-FIFO queue initially at state $\bz$ is
\begin{equation}\label{eq:slowdown-fifo}
\cost_\bz(x) = 1 + \frac{u_{\bz}}{x} + \frac{\lambda\cdot\E{\JS^{-1}}}{2(1-\rho)}(2u_{\bz}x+x^2).
\end{equation}
\end{corollary}
In both cases, it is implicitly assumed that the 
question about the admittance
is a one-time operation and
the future behavior of the queue is according to the standard M/G/1-FIFO.
The proofs follow trivially from \eqref{eq:cost} and \eqref{eq:vn-fifo}.

\subsection{M/G/1-LIFO Queue}

Next we derive the size-aware value function for the LIFO scheduling discipline.
To this end, 
let us denote the state of the system with $n$ jobs
by a vector $\bz = ((\Delta_1,\hc_1); \ldots; (\Delta_n,\hc_n))$,
where $\Delta_i$ denotes the remaining service requirement (measured in time)
and $\hc_i$ the \emph{holding cost rate} of job $i$. 
Job $1$ is the latest arrival and currently receiving service (if any),
i.e., without new arrivals the jobs are processed in the natural order:
$1, 2, \ldots, n$.

\begin{prop}%
For the size-aware relative value in a preemptive M/G/1-LIFO queue with respect
to arbitrary job specific holding cost it holds that
\begin{equation}\label{eq:vn-plifo}
v_{\bz}-v_0 = \frac{1}{1-\rho}\sum_{i=1}^{n} \left( \hc_i \sum_{j=1}^i \Delta_j\right).
\end{equation}
\end{prop}
\begin{proof}
We compare System 1 initially in state $\bz$ and System 2 initially empty.
The current state bears no meaning for the future arrivals with preemptive LIFO,
and thus the difference in the expected costs is equal to the cost
the current $n$ jobs in System 1 incur:
$$
v_{\bz}-v_0 = \sum_{i=1}^{n} \hc_i\, \E{ R_i},
$$
where $R_i$ denotes the remaining sojourn time of job $i$ \cite{hyytia-ejor-2012},
$$
\E{R_i} = \frac{\sum_{j=1}^i \Delta_j}{1-\rho},
$$
which completes the proof.
\end{proof}
For the mean sojourn time, the holding cost rate is constant $\hc_i=1$ and
a sufficient state description is $(\Delta_1,\ldots,\Delta_n)$.
\begin{corollary}%
The cost in terms of sojourn time due to accepting a job with size $x$ to 
a size-aware preemptive M/G/1-LIFO queue %
at state $\bz$ is
$$
\cost_{\bz}(x) = \frac{1}{1-\rho} (n+1) x.
$$
\end{corollary}
For the slowdown, the initial service requirements
define the holding costs and a sufficient
state description is $\bz =((\Delta_1,\Delta_1^*),..,(\Delta_1,\Delta_1^*))$,
where $\Delta_i$ and $\Delta_i^*$ denote the remaining and initial service requirement of job $i$,
so that the holding cost of job $i$ is $\hc_i=1/\Delta_i^*$:
\begin{corollary}%
The cost in terms of slowdown due to accepting a job with size $x$ to 
a size-aware preemptive M/G/1-LIFO queue %
at state $\bz$ is
\begin{equation}\label{eq:slowdown-lifo}
\cost_{\bz}(x) = \frac{1}{1-\rho}\left( 1 + \sum_{i=1}^{n} x/\Delta_i^*\right).
\end{equation}
\end{corollary}
The proofs follow trivially from \eqref{eq:vn-plifo} and from definition \eqref{eq:cost}.
Note that $\cost_\bz(x)$ with (preemptive) LIFO does not depend
on the remaining service times $\Delta_i$ due to the preemption.

\subsection{M/G/1-SPTP Queue}

Next we derive the corresponding size-aware value function for the SPTP scheduling discipline.
The state description for SPTP is
$\bz=((\Delta_1,\Delta_1^*,\hc_1),$  $\ldots,$ $(\Delta_n,\Delta_n^*,\hc_n))$,
where,
without loss of generality,
we assume a decreasing priority order,
$
\Delta_i\Delta_i^* < \Delta_{i+1}\Delta_{i+1}^*,
$
i.e., the job $1$ (if any) is currently receiving service.
Then, $u_{\bz}(h)$ denotes the amount of work with a higher priority than $h$,
$$
u_{\bz}(h) \triangleq \sum_{i:\Delta_i\Delta_i^* < h} \Delta_i.
$$
Let $\lambda(x)$, $m(x)$ and $\rho(x)$ denote the
arrival rate,
mean job size
and
offered load due to jobs shorter than $x$,
\begin{equation}\label{eq:rho-x}
  \begin{array}{ll}
    \lambda(x) &\triangleq \lambda\, \PP{X<x},\\
    m(x)       &\triangleq \cE{X}{X<x},\\
    \rho(x)    &\triangleq \lambda \int_0^x  t\,f(t)\,dt,
  \end{array}
\end{equation}
where $f(x)$ denotes the job size pdf.
\begin{lemma}
The mean remaining sojourn time of 
a $(\Delta,\Delta^*)$-job in an M/G/1-SPTP queue initially in state $\bz$
is given by
\begin{equation}%
\label{eq:ER-sptp}
\E{R_{\bz}(\Delta,\Delta^*)} = \frac{u_{\bz}(\Delta\Delta^*)}{1{-}\rho(\sqrt{\Delta\Delta^*})} +
\frac{2}{\Delta^*} \int\limits_0^{\sqrt{\Delta\Delta^*}} \frac{x\,dx}{1{-}\rho(x)}.
\end{equation}
\end{lemma}
\begin{figure}
  \centering
  \includegraphics[height=26mm]{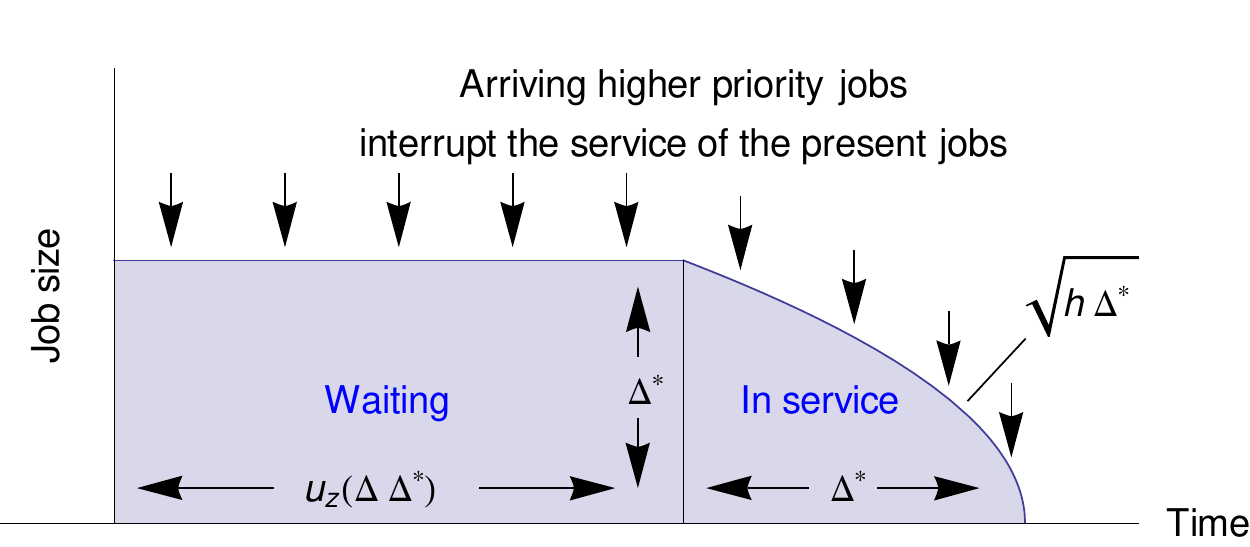}
  \caption{Remaining sojourn time of an $(\Delta,\Delta^*)$-job
    in an M/G/1-SPTP queue.}
\label{fig:ER-sptp}
\end{figure}
\begin{proof}
Let $k$ denote the $(\Delta,\Delta^*)$-job,
whose remaining sojourn time depends
on the initial and later arriving
higher priority work.
We can assume that the latter
are served immediately according to LIFO,
thus triggering mini busy periods.
Let $h=h(t)$ denote the remaining service time of job $k$
at (virtual) time $t$, where we have omitted these mini busy periods
from the time axis.
During $0<t<u_{\bz}(\Delta\Delta^*)$,
job $k$ is waiting and $h=\Delta$,
but as the service begins $h \to 0$ linearly.
The later arriving higher priority
jobs shorter than $\sqrt{h\Delta^*}$ constitute 
an inhomogeneous Poisson process with rate $\lambda(\sqrt{h\Delta^*})$,
as illustrated in Fig.~\ref{fig:ER-sptp}.
The mean duration of a mini busy period is (cf.\ the mean busy period in M/G/1),
$$
D(h) \triangleq \frac{ m(\sqrt{h\,\Delta^*}) }{1-\rho(\sqrt{h\,\Delta^*})}.
$$
The mean waiting 
time before job $k$ receives service 
for the first time is
$u_{\bz}(\Delta\Delta^*) +  \lambda(\sqrt{h\Delta^*})\,u_{\bz}(\Delta\Delta^*)\cdot D(\Delta)$,
which gives the first term in \eqref{eq:ER-sptp}.

The service time $\Delta$ and the additional delays %
during the service %
are on average 
$$
\Delta +
\int_0^{\Delta} \lambda(\sqrt{h\,\Delta^*}) \cdot D(h)\,dh;
$$
refer to the ``in service'' region of Fig.~\ref{fig:ER-sptp}.
Change of integration variable, $x=\sqrt{h\,\Delta^*}$ 
then gives the second term in \eqref{eq:ER-sptp}, which completes the proof.
\end{proof}

\begin{figure*}[th!]
\centering
\includegraphics[width=133mm]{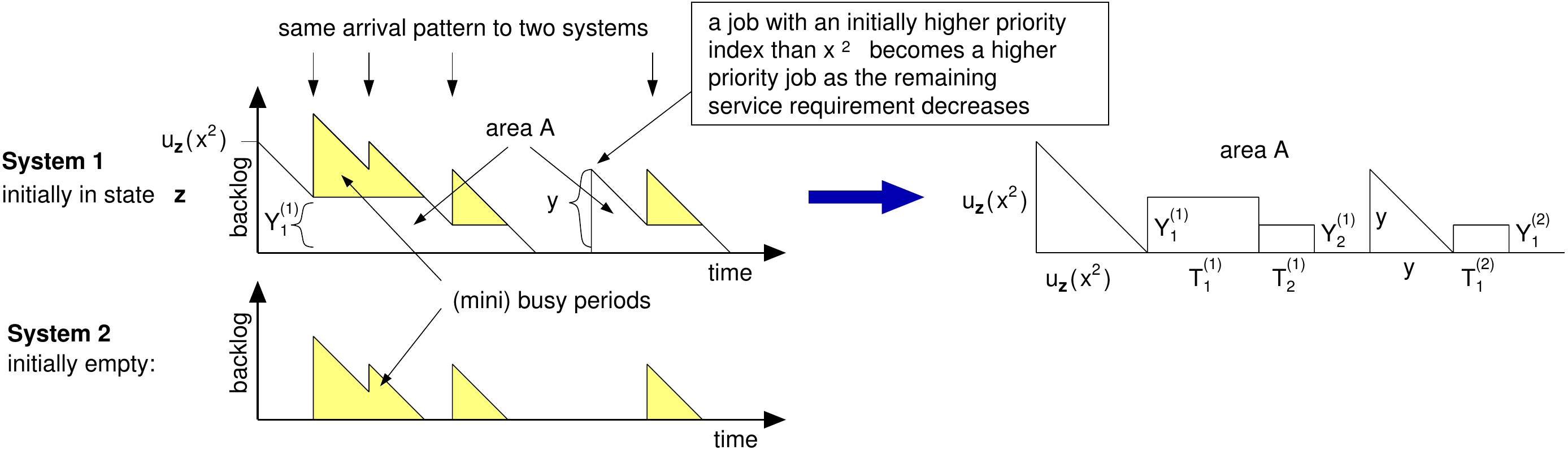}
\caption{Derivation of the value function for an M/G/1-SPTP queue with respect to slowdown.}
\label{fig:vn-sptp}
\end{figure*}
Even though SPTP explicitly tries to minimize the mean slowdown,
we derive next a general expression for the value function
with arbitrary job specific holding costs:
\begin{prop}%
For the size-aware relative value 
in an M/G/1-SPTP queue with respect to arbitrary
job specific holding costs 
it holds that
\begin{equation}%
\label{eq:vn-sptp}
\begin{array}{l@{\,}l}
\dps
v_{\bz}{-}v_0 &=\dps
 \sum_{i=1}^n \hc_i \biggl(
  \frac{\sum_{j=1}^{i-1} \Delta_j}{1{-}\rho(\tilde\Delta_i)} %
  + \frac{2}{\Delta_i^*} %
  \int_0^{\tilde\Delta_i} %
  \frac{x\,dx}{1{-}\rho(x)}\biggr)
  \\[4mm]
&\dps\;+\;
\frac{\lambda}{2} \sum_{i=0}^n
\biggl[
(\sum_{j=1}^i \Delta_j)^2 \int_{\tilde\Delta_i}^{\tilde\Delta_{i+1}} \frac{\hc(x)\,f(x)}{(1-\rho(x))^2}\,dx
 \; +
\\[4mm]
&\dps\qquad\quad
\sum_{j=i+1}^n (\Delta_j^*)^{-2} \int_{\tilde\Delta_i}^{\tilde\Delta_{i+1}} \frac{x^4\,\hc(x)\,f(x)}{(1-\rho(x))^2}\,dx
\biggr]
\end{array}
\end{equation}
where $\hc(x) \triangleq \cE{\HC}{\JS=x}$ and %
$$
\tilde\Delta_i = \left\{
\begin{array}{ll}
0, & i=0\\
\sqrt{\Delta_i\Delta_i^*}, & i=1,\ldots,n\\
\infty, & i=n+1.
\end{array}\right.
$$
\end{prop}
\begin{proof}
The relative value comprises the mean holding cost 
$h_1$ accrued by the current $n$ tasks in state $\bz$,
$$\bz=((\Delta_1,\Delta_1^*,\hc_1), \ldots, (\Delta_n,\Delta_n^*,\hc_n)),$$ 
and the difference in costs accrued by the later arriving jobs, $h_2$,
$v_{\bz}-v_0 = h_1+h_2$.
The first summation over $n$ gives $h_1$,
i.e., it follows from
multiplying \eqref{eq:ER-sptp} with the job specific holding cost $\hc_i$
and adding over all $i$.

The latter integral corresponds to 
$h_2$
where the accrued costs are conditioned on the size $x$ of an arriving job.
The corresponding arrival rate is $\lambda\,f(x)\,dx$. 
The mean difference in sojourn time experienced by an arriving job with size $x$
is accrued during the initial waiting time -- once a job enters the service for the first time
its mean remaining sojourn time is the same in both systems.
The initial waiting time is a function of higher priority workload upon arrival,
denoted by $U_{\bz}(x^2,t)$ for an initial state $\bz$ at arrival time $t$.
In particular, the mean (initial) waiting time is simply
$\E{U_{\bz}(x^2,t)}/(1-\rho(x))$,
which gives
$$
h_2 = \lambda \int_{0}^{\infty} \frac{f(x)}{1-\rho(x)} \,(\E{U_{\bz}(x^2,t)-U_{0}(x^2,t)})\,dx.
$$
Next we refer to Fig.~\ref{fig:vn-sptp} and observe that %
$U_{\bz}(x^2,t)-U_{0}(x^2,t)$
in the integrand corresponds to area A. %
This area consists of one $u_{\bz}(x^2)$ triangle, followed by
$N_0$ rectangles, and similar sequences starting with 
a $x^2/\Delta_i^*$ triangle for $\{i\,:\,\Delta_i\Delta_i^*>x^2\}$,
each followed by $N_i$ rectangles.
The $N_i$ are random variables corresponding to the number
of mini busy periods during the service time of a particular triangle.
The first triangle corresponds to the initially higher priority workload $u_{\bz}(x^2)$,
and the latter %
to the jobs with
initially lower priority, i.e., to jobs $i$ with $\Delta_i\Delta_i^*>x^2$.
At some point in time, 
the remaining service requirement of such a job $i$
decreases to $y_i=x^2/\Delta_i^*$ and its priority index drops below $x^2$.
For a sequence starting with a $y$-triangle,
the number of mini busy periods (rectangles) obeys Poisson distribution with 
mean %
$\lambda(x)\,y$.
The height of a rectangle is uniformly distributed in $(0,y)$ having the mean $y/2$
(property of Poisson process), while the width corresponds to the duration of
a busy period in a work conserving M/G/1 queue having the mean $m(x)/(1-\rho(x))$.
Thus, the mean total area is
\begin{align*}
\frac{y^2}{2} + \lambda(x) y \cdot \frac{y}{2} \cdot \frac{m(x)}{1-\rho(x)}
  = \frac{y^2}{2(1-\rho(x))}.
\end{align*}
Therefore,
\begin{align*}
h_2 &=
\frac{\lambda}{2} \int\limits_{0}^{\infty} \frac{\hc(x)\,f(x)}{(1-\rho(x))^2}\,\biggl(
u_{\bz}(x^2)^2 + 
 \sum_{i:\Delta_i\Delta_i^*>x^2} \left(\frac{x^2}{\Delta_i^*}\right)^2\biggr)\,dx,
\end{align*}
where %
$\hc(x) = \cE{\HC}{\JS=x}$ factor in the numerator corresponds
to the mean holding cost of an $x$-job.
For each interval $x \in (\tilde\Delta_i,\tilde\Delta_{i+1})$,
$u_{\bz}(x^2) = \sum_{j=1}^{i-1} \Delta_j$ and
$$
\sum_{j:\Delta_j\Delta_j^*>x^2} \left(\frac{x^2}{\Delta_j^*}\right)^2
= x^4 \sum_{j=i+1}^n (\Delta_j^*)^{-2},
$$
and integration in $n+1$ parts completes the proof.
\end{proof}

Recall that with respect to the mean sojourn time, $\hc_i=1$ and $\hc(x)=1$, while
for the slowdown criterion, $\hc_i=1/\Delta_i^*$ and $\hc(x)=1/x$.
Hence, \eqref{eq:vn-sptp} allows one to compute
the mean cost $\cost_{\bz}(x)$
due to accepting a given job in terms of sojourn time or slowdown.
We omit the explicit expression for $\cost_{\bz}(x)$
for brevity. %

The corresponding value functions for SPT and SRPT are derived in the Appendix.
The numerical evaluation of the value function 
for SPTP, SPT and SRPT
is not as unattractive as it first may seem.
Basically one needs to be able to compute integrals of form
$$
F_k(x)=\int_{0}^{x} \frac{t^k\,\hc(t)\,f(t)}{(1-\rho(t))^2}\,dt,
\text{ and }
G_d(x)=\int_{0}^{x} \frac{t^d}{1-\rho(t)}\,dt,
$$
where $k=0,2,4$ and $d=0,1$
depending on the discipline. %
Thus, 
e.g., a suitable interpolation of the $F_k(x)$
and $G_d(x)$
enables on-line computation of the value function.

\section{Dispatching Problem}
\label{sec:dispatching}

Next we utilize the results of Section~\ref{sec:values}
in the dispatching problem with parallel servers
and focus solely on minimizing the mean slowdown.
The dispatching system 
illustrated in Fig.~\ref{fig:dispatch-m}
comprises $m$ servers with 
service rates $\nu_1,\ldots,\nu_m$.
Jobs arrive according to a Poisson process with rate $\lambda$ and
their service requirements are i.i.d.\ random variables with
a general distribution.
The jobs are served according to a given scheduling discipline (e.g., FIFO) in each server.

The slowdown for an isolated queue was defined as the sojourn time $\ST$
divided by the service requirement $\JS$ \cite{harchol-balter-peva-2002} (both measured in time),
$\gamma = \ST/\JS$.
However, as we consider heterogeneous servers with rates $\nu_i$,
the service requirement $\JS$ of size $Y$ job (measured, e.g., in bytes)
is no longer unambiguous but a server specific quantity,
$$
\JS_i = Y / \nu_i.
$$
Therefore, for a dispatching system we compare the sojourn time $\ST$ to the hypothetical
service time if all capacity could be assigned to process a given job \cite{downey-hpdc-1997},
$$
\gamma^* \triangleq \frac{\ST}{Y / \sum_i \nu_i}.
$$
The relationship between the queue $i$ specific slowdown $\gamma$
and the system wide slowdown $\gamma^*$ is
$
\gamma^* = \gamma \cdot (\sum_j \nu_j)/\nu_i$.

\subsection{Random Dispatching Policies}
The so-called
Bernoulli splitting assigning jobs independently
in random using probability distribution $(p_1,\ldots,p_m)$
offers a good state-independent basic dispatching policy. %
Due to Poisson arrivals, each queue also receives
jobs according to a Poisson process with rate $p_i\,\lambda$.

\begin{defn}[RND-$\rho$]
The \emph{RND-$\rho$ dispatching policy}
balances the load equally by setting $p_i = \nu_i / \sum_j \nu_j$.
\end{defn}

As an example, the mean slowdown in a preemptive LIFO 
or PS queue with RND-$\rho$ 
policy is, %
$$
\E{\gamma^*}= \bigl(\sum_j \nu_j\bigr) \frac{m}{\sum_j \nu_j - \lambda\,\E{Y}},
$$
where the denominator corresponds to the excess capacity.
Thus, with RND-$\rho$ and LIFO queues,
the slowdown criterion is $m$ times
higher when $m$ servers are used instead of a single fast one
irrespectively of the service rate distribution $\nu_i$.

For identical servers, $\nu_1=\nu_2=\ldots=\nu_m$,
the RND-$\rho$ dispatching policy reduces to RND-U:

\begin{defn}[RND-U]
The RND-U dispatching policy assigns jobs
randomly in uniform using $p_i = 1/m$.
\end{defn}

RND-U is obviously the optimal random policy in case of identical servers.

In general, the optimal splitting probabilities depend on the service time
distribution and the scheduling discipline.
For the preemptive LIFO (and PS) server systems
the mean slowdown 
with a random dispatching policy is given by
$$
\E{\gamma^*}= \bigl(\sum_j \nu_j\bigr)\sum_{i=1}^m \frac{p_i}{\nu_i-p_i\cdot\lambda\,\E{Y}},
$$
where $\sum_j \nu_j$ is a constant that can be neglected
when optimizing the $p_i$.
\begin{defn}[RND-opt]
The optimal %
random dispatching policy for LIFO/PS queues,
referred to as RND-opt, splits the incoming tasks
using the probability distribution,
$$
p_i = \frac{\nu_i-\sqrt{\nu_i}\,G}{\lambda\,\E{Y}},
\quad
\text{where }
G = \frac{\sum_i \nu_i - \lambda\,\E{Y}}{\sum_i \sqrt{\nu_i}}.
$$
\end{defn}
The result is easy to show with the aid of Lagrange multipliers.
We note that when some servers are too slow, the above
gives infeasible values for some $p_i$, 
in case of which one simply excludes
the slowest server from the solution and re-computes a new probability
distribution. This is repeated until a feasible solution is found.
A more explicit formulation is given in \cite{bell-management-1983,haviv-orl-2007}.
RND-opt is insensitive to the job size distribution
and optimal only for LIFO and PS. %

Pollaczek-Khinchin mean value formula enables one to analyze FIFO queues
and, e.g., to write an expression for the mean slowdown with RND-$\rho$.
Also the optimal splitting probabilities can
be computed numerically for an arbitrary job size distribution.
According to \eqref{eq:fifo-slowdown}, the mean slowdown in each
queue depends on the mean waiting time $\E{W_i}$ and $\E{\JS_i^{-1}}$,
where the latter is assumed to exist and to be finite.
Hence, the optimal probability distribution with respect to slowdown
is the same as with the mean sojourn time.

Similarly, the mean sojourn time in SPT and SRPT queues is known
which allows a numerical optimization of the splitting probabilities.
For simplicity, we consider only RND-$\rho$ and RND-opt in this paper
as compact
closed form expressions for the $p_i$ are only available for these.

\subsection{SITA-E Dispatching Policy}

With FIFO queues, a state-independent dispatching
policy known as the size-interval-task-assignment (SITA)
has proven to be efficient especially with heavy-tailed
job size distributions \cite{crovella-sigmetrics-1998,harchol-balter-pdc-1999,bachmat-peva-2009}.
The motivation behind SITA is to segregate the long jobs from the short ones.
Reality, however, is more complicated and 
segregating the jobs categorically can also
give suboptimal results
\cite{harchol-balter-sigmetrics-2009}.
Here we assume $\nu_1 \ge \nu_2 \ge \ldots \ge \nu_m$
and a continuous job size distribution with pdf $f(x)$.
\begin{defn}[SITA]
A SITA policy is
defined by disjoint job size intervals
$\{(\xi_0,\xi_1]$, $(\xi_1,\xi_2]$, $\ldots$, $(\xi_{m-1},\xi_m]\}$,
and assigns a job with size $x$ to server $i$ iff
$x \in (\xi_{i-1},\xi_i]$.
\end{defn}

Without loss of generality, one can assume that $\xi_0=0$ and $\xi_m=\infty$.
Note that in contrast to random policies, SITA 
assumes that the dispatcher is aware of the size of the new job.
In this paper, we limit ourselves to SITA-E, where E
stands for \emph{equal load}. That is, the size intervals
are chosen in such a way that the load is balanced %
between the servers. %

\begin{defn}[SITA-E]
With SITA-E dispatching policy the thresholds $\xi_i$
are defined in such a way that
$$
\frac{1}{\nu_i} \int_{\xi_{i-1}}^{\xi_i} x\,f(x)\,dx
=
\frac{1}{\nu_j} \int_{\xi_{j-1}}^{\xi_j} x\,f(x)\,dx,
\qquad
\forall\;i,j.
$$
\end{defn}

Thus, similarly as RND-$\rho$,
also the SITA-E policy is insensitive to the arrival rate $\lambda$.

\subsection{Improved Dispatching Policies}

The important property the above state-independent dispatching policies
have is that the arrival process to each queue 
is a Poisson process.
Consequently, 
the value functions derived earlier allow us to quantify
the value function of the whole system.
With a slight abuse of notation,
$$
v_{\bZ}-v_{\mathbf{0}} = \sum_{i=1}^m (v_{\bz_i}^{(i)}-v_0^{(i)}),
$$
where 
$v_{\bz_i}^{(i)}$ denotes the relative value of queue $i$ in state $\bz_i$,
and $\bZ=(\bz_1,\ldots,\bz_m)$ the state of the
whole system.

\subsubsection*{Policy Improvement by Role Switching}
One interesting opportunity to utilize the relative values 
is to \emph{switch} the roles of two queues \cite{hyytia-ejor-2012}.
Namely, with an arbitrary state-independent policy
one can switch the input processes of any two identical servers at any moment,
and effectively end up to a new state of the same system. 
Despite its limitation (identical servers), 
\emph{switching} is an interesting policy improvement method
that requires little additional computation.
The switch should only
be carried out when the new state has
lower expected future costs, i.e.\ when
for some $i\ne j$, the rates are equal, $\nu_i=\nu_j$, and
$$
v_{(\bz_1,..,\bz_i,..\bz_j,..,\bz_m)}
<
v_{(\bz_1,..,\bz_j,..\bz_i,..,\bz_m)}.
$$
Carrying out this operation whenever a new job
has arrived to the system %
reduces the state-space by \emph{removing} such states
for which a better alternative to continue exists.
Formally,
let $\pi(\bZ)$ denote the set of feasible permutations
of the queues' roles, where the input processes between
identical servers are switched. %
Then, the optimal state-space reduction,
implicitly defining a new policy,
is given by
$$
\bZ \leftarrow \argmin{\bZ' \in \pi(\bZ)}\; v_{\bZ'}.
$$

To elaborate this, consider a RND-$\rho$/LIFO system.
In this case,
\eqref{eq:vn-plifo} implies that switching
the roles of any two identical queues makes
no difference to the relative value.

In contrast, with FIFO discipline the interesting quantity
from \eqref{eq:vn-fifo} for two identical servers $i$ and $j$ is identified to be
$$
C_{ij} \triangleq \frac{\lambda_i\,\E{\HC_i}}{1-\rho_i} \cdot u_{\bz_j}^2.
$$
Switching the input processes between queues $i$ and $j$ 
leads to a state with a lower relative value if
$C_{ii} + C_{jj} > C_{ij} + C_{ji}$.
Again, with RND-$\rho$ the factor 
$\lambda_i\,\E{\HC_i} / (1-\rho_i)$ is a constant
for any two identical servers and switching provides no gain.
However, with SITA-E, %
even though the denominator $1-\rho_i$ is
a constant, the numerator is not.
Let $\JS_i$ denote a job size in queue $i$.
Then $\E{\JS_i} < \E{\JS_{i+1}}$,
and therefore $\lambda_i > \lambda_{i+1}$.
Moreover, $\E{1/\JS_i} > \E{1/\JS_{i+1}}$,
yielding
$$
\lambda_i\,\E{1/\JS_i} > \lambda_{i+1}\,\E{1/\JS_{i+1}}.
$$
Consequently, with the \emph{SITA-E with switch} policy,
the optimal permutation, after inserting a new job to a queue,
is the one with increasing backlogs for identical servers:
\begin{defn}[SITA-Es]
SITA-E with switch dispatching policy behaves similarly as SITA-E with a distinction
that after each task assignment
the queues are permutated in such a way that
$u_{\bz_{i}} \le u_{\bz_{i+1}}$ for all identical servers $i$ and $i+1$.
\end{defn}

\subsubsection*{Policy Improvement by FPI}
The \emph{first-policy-iteration} (FPI)  %
is a general method of the MDP framework to improve any given policy.
We apply %
it to the dispatching problem
\cite{hyytia-ejor-2012}. 
Suppose that the scheduling discipline in each queue is fixed
and that a state-independent basic dispatching policy would assign
a new job to some queue.
Given the relative values and the expected cost associated 
with accepting the job to each queue,
we can carry out FPI:
we deviate from the default action
if the expected cost is smaller with some other action,
thereby decreasing the expected cumulative costs in infinite time horizon. 

Let $\lambda_i$ 
denote the arrival rate to queue $i$ according to the basic dispatching policy,
and $Y_i$ the corresponding job size. %
With a state-independent policy,
accepting a job to queue $i$ 
does not affect the future behavior of the other queues. %
Thus, the cost of assigning a job with size $y$ to queue $i$ is
$$
\cost_{\bZ}(y,i) = \cost^*_{\bz_i}(y).
$$
where $\cost^*_{\bz_i}(y)$ is the mean admittance cost of a job with size $y$ (measured, e.g., in bytes)
to queue $i$, where its service requirement (measured in time) would be $y/\nu_i$.
For slowdown, we have an elementary relation
$$
\cost^*_{\bz_i}(y) = \frac{\sum_j \nu_j}{\nu_i} \cdot \cost_{\bz_i}(y/\nu_i).
$$
For \emph{FIFO queues}, \eqref{eq:slowdown-fifo} gives
$$
\cost^*_{\bz_i}(y) = \bigl(\sum_j \nu_j\bigr)\left(
\frac{1}{\nu_i} {+} \frac{u_{\bz_i}}{y} {+} \frac{\lambda_i\,\E{Y_i^{-1}}}{2}\cdot\frac{2 u_{\bz_i}y + y^2/\nu_i}{\nu_i-\lambda_i\,\E{Y_i}}
\right).
$$
For \emph{preemptive LIFO} queues, %
\eqref{eq:slowdown-lifo} similarly gives
$$
\cost^*_{\bz_i}(y) = \bigl(\sum_j \nu_j\bigr) \frac{1+(y/\nu_i)\sum_{j=1}^{n_i} 1/\Delta^*_{i,j}}{\nu_i-\lambda_i\,\E{Y_i}},
$$
where $n_i$ denotes the number of jobs in queue $i$ and
$\Delta^*_{i,j}$ the initial service requirement (in time) of job $j$ in queue $i$.
According to the FPI principle, one
simply chooses the queue with the smallest (expected) cost,
\begin{equation}\label{eq:fpi}
\alpha_{\bZ}(y) \triangleq \argmin{i} \;\cost^*_{\bz_i}(y).
\end{equation}
We refer to these improved dispatching policies simply as the \emph{FPI-$p$ policy},
where $p$ denotes the basic dispatching policy, e.g., RND-opt or RND-$\rho$,
where for brevity reasons the RND prefix is often omitted.
Note that the factor $\sum_j \nu_j$ in both expressions
is a common constant for all queues.

\begin{table*}[th]
\centering
\small
\begin{tabular}{|lp{130mm}|}
\hline
&\\[-2mm]
\poli{RND}-$\rho$ & state-independent random policy with load balancing\\
\poli{RND-opt}    & state-independent random policy minimizing the mean slowdown (assumes LIFO/PS)\\
\poli{SITA-E}     & state-independent size-interval-task-assignment with equal loads (optionally, with switch)\\
\hline
&\\[-2mm]
\poli{Round-Robin} & assigns arriving tasks sequentially to servers \cite{ephremides-tac-1980}\\
\poli{LWL$^-$}        & least-work-left, assigns a job to server with the least amount unfinished work upon arrival \\
\poli{LWL$^+$}        & same as \poli{LWL$^-$} but based on the unfinished work (in time) including the new job \cite{hyytia-ejor-2012}\\
\poli{JSQ}         & join-the-shortest-queue, i.e., the one with the least number of jobs \cite{winston-applied-1977}\\
\poli{Myopic}      & minimize the mean slowdown on condition that no further jobs arrive\\
\poli{FPI}         & first policy iteration on the state-independent RND-$\rho$, RND-opt and SITA-E policies\\
\hline
\end{tabular}
\caption{Dispatching policies evaluated in the numerical examples.}
\label{tbl:policies}
\end{table*}

In a symmetric case of $m$ identical servers and the RND-U basic policy,
a corresponding FPI-based policy reduces to a well-known dispatching policy,
i.e.,
with FIFO queues, FPI yields \poli{LWL},
and with LIFO, the \poli{Myopic} policy (see Table~\ref{tbl:policies}).
Moreover, given additionally a constant job size, then in case of LIFO
queues one ends up with the JSQ policy.
In general, for constant job sizes the slowdown objective
reduces to the minimization of the mean sojourn time.

\begin{figure}
\centering
\begin{tabular}{cc}
\parbox{36mm}{\includegraphics[width=36mm]{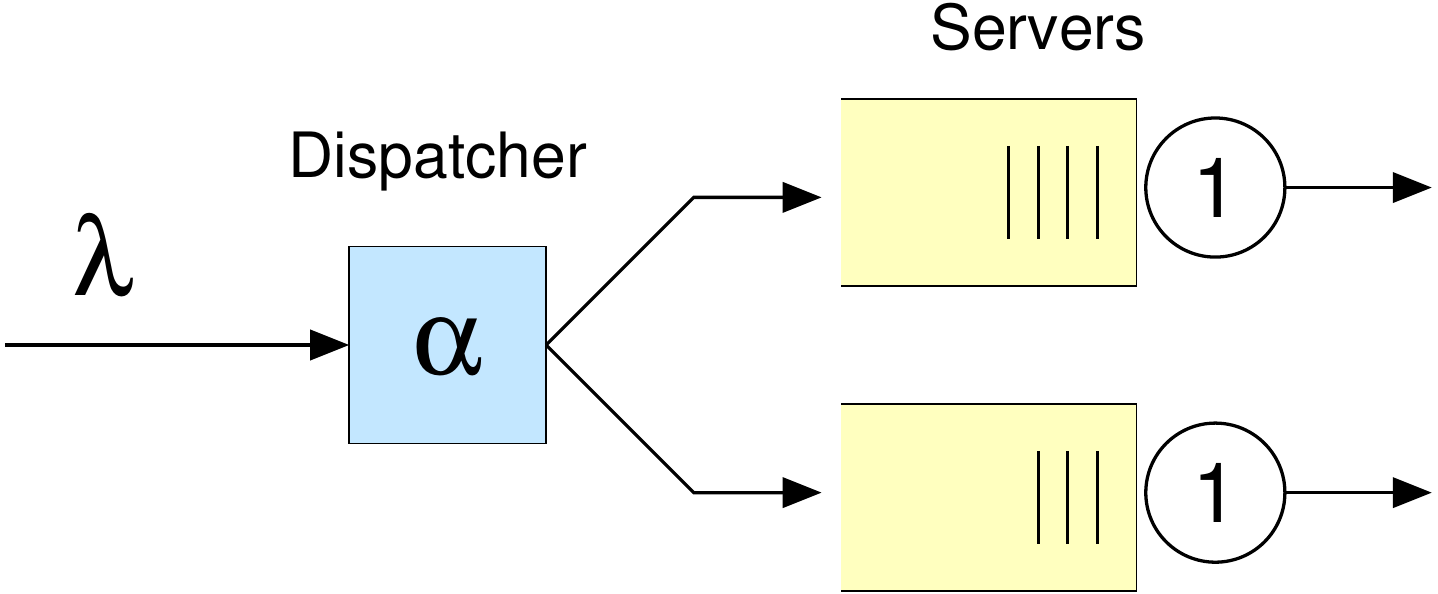}} &
\parbox{36mm}{\includegraphics[width=36mm]{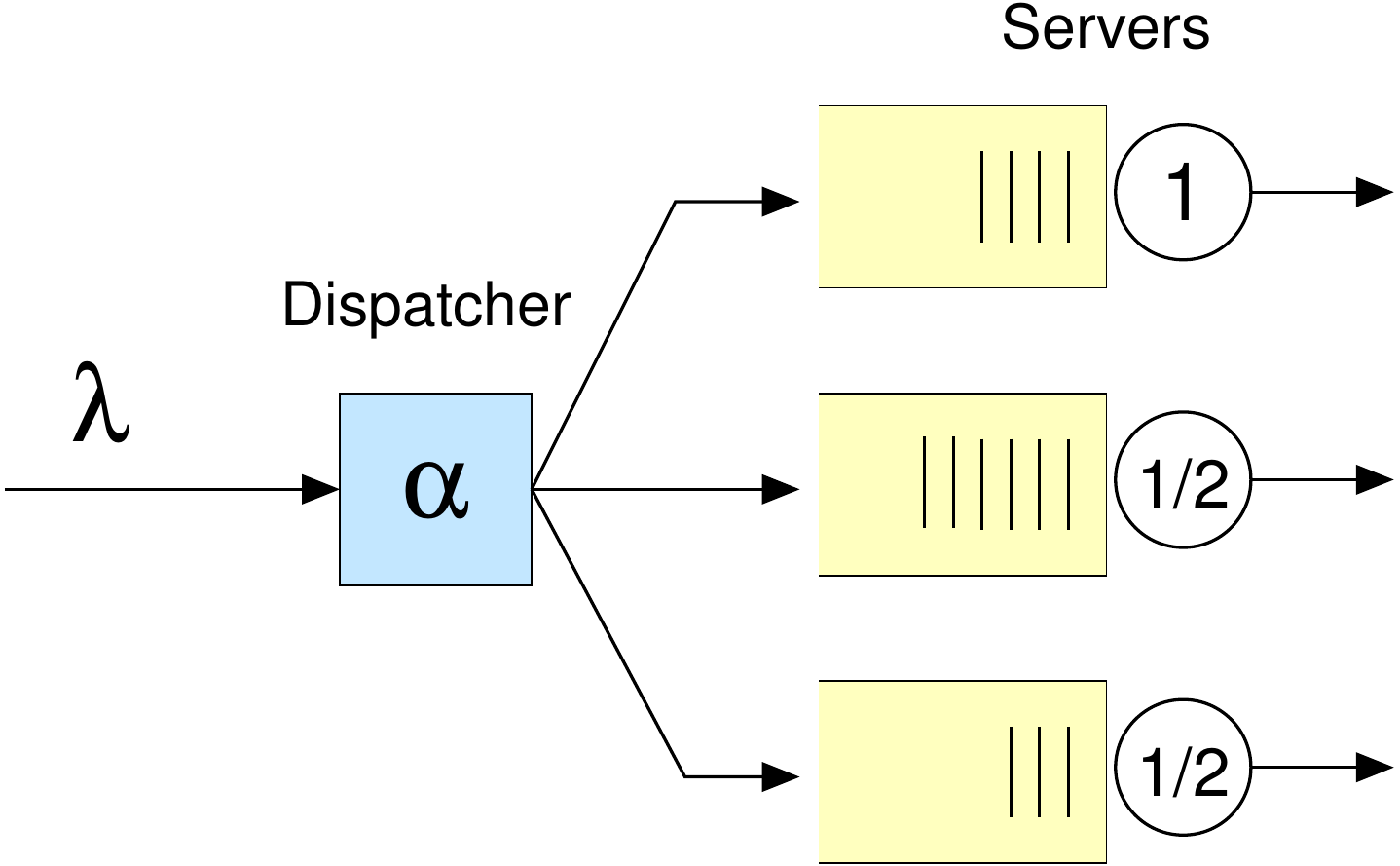}} \\[12mm]
\small (a) Two identical servers &
\small (b) Three heterogeneous servers
\end{tabular}
\caption{Example dispatching systems.}
\label{fig:dispatch-3}
\end{figure}

\section{Numerical Examples}
\label{sec:examples}

Let us next evaluate by means of numerical simulations
the improved dispatching policies derived in Section~\ref{sec:dispatching}.
To this end, we have chosen to consider two elementary server systems that are illustrated in Fig~\ref{fig:dispatch-3}:
\begin{itemize}\parskip-0pt
\item[(a)] Two identical servers with rates $\nu_1=\nu_2=1$,
\item[(b)] Three heterogeneous servers with rates $\nu_1=1$ and $\nu_2=\nu_3=1/2$.
\end{itemize}
Thus, the total service rate in both systems is equal to $2$.
We compare the mean slowdown performance of the FPI policies 
against several
well-known heuristic dispatching policies,
including 
the state-independent policies
\poli{RND}-opt, \poli{RND}-$\rho$ and \poli{SITA-E}.
The somewhat more sophisticated least-work-left (LWL) policies
choose the queue with the least amount of unfinished work (backlog).
The difference between the LWL policies is 
that \poli{LWL$^-$} considers the situation without the new job,
and \poli{LWL$^+$} afterwards. With identical servers the LWL policies are equivalent.
\poli{JSQ} chooses the queue with the least number of jobs.
With the \poli{LWL} and \poli{JSQ}, the ties are broken in favor of a faster server.
The \poli{Myopic} policy assumes in a greedy fashion
that no further jobs arrive and chooses the queue which minimizes
the immediate cost the known jobs accrue.
All policies are listed in Table~\ref{tbl:policies}.

In many real-life settings, the job sizes
have been found to exhibit a heavy-tailed behavior,
(cf., e.g., file sizes in the Internet).
Thus, in most examples we assume a \emph{bounded Pareto distribution} with pdf
\begin{equation}\label{eq:pareto}
f(x) = \frac{\alpha k^\alpha}{1-(k/p)^\alpha}\,x^{-\alpha-1},
\quad k \le x \le p,
\end{equation}
where $(k,p,\alpha)\,{=}\,(0.33959,1000,1.5)$
so that $\E{X} \,{\approx}\, 1$.
This job size distribution is particularly suitable for SITA-E.

\begin{figure*}
\centering
\includegraphics[width=70mm]{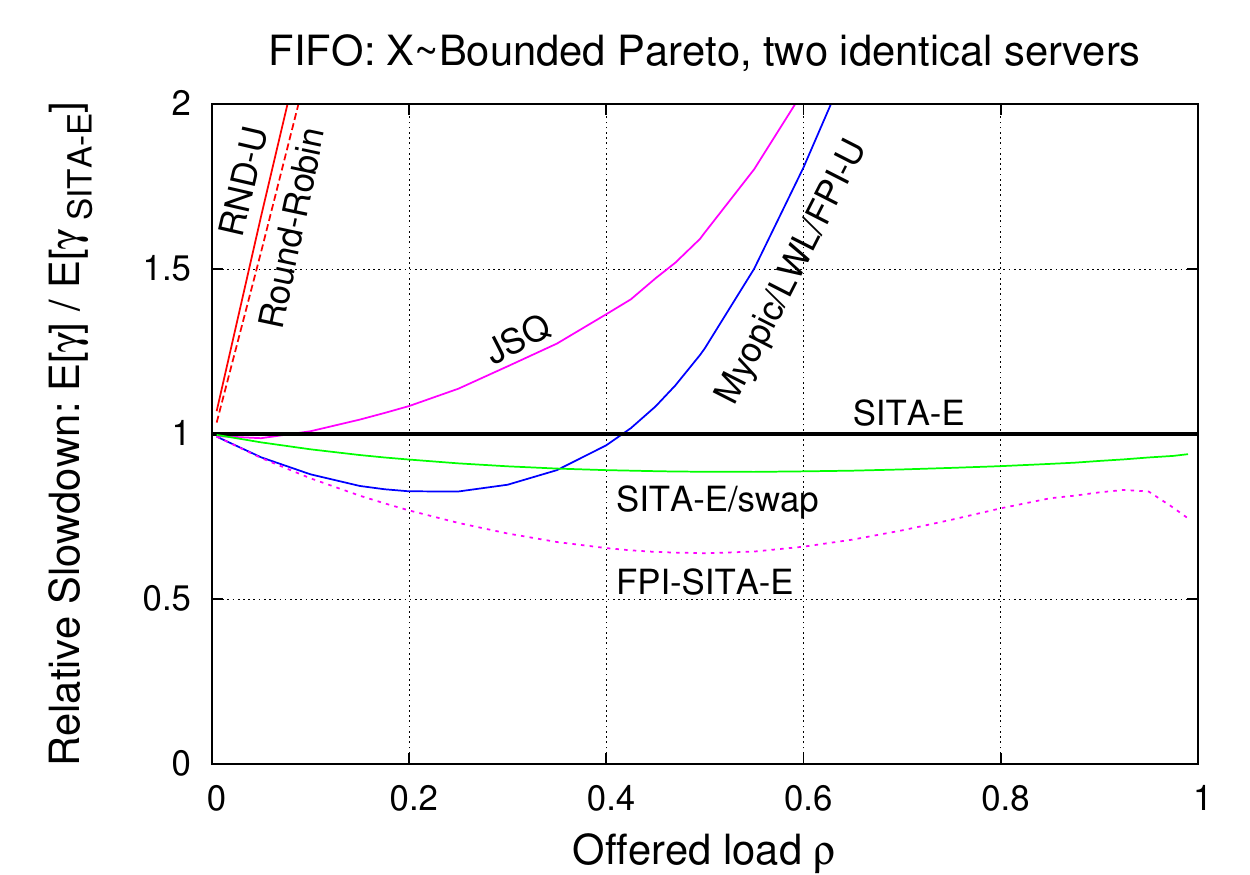}
\hspace{20mm}
\includegraphics[width=70mm]{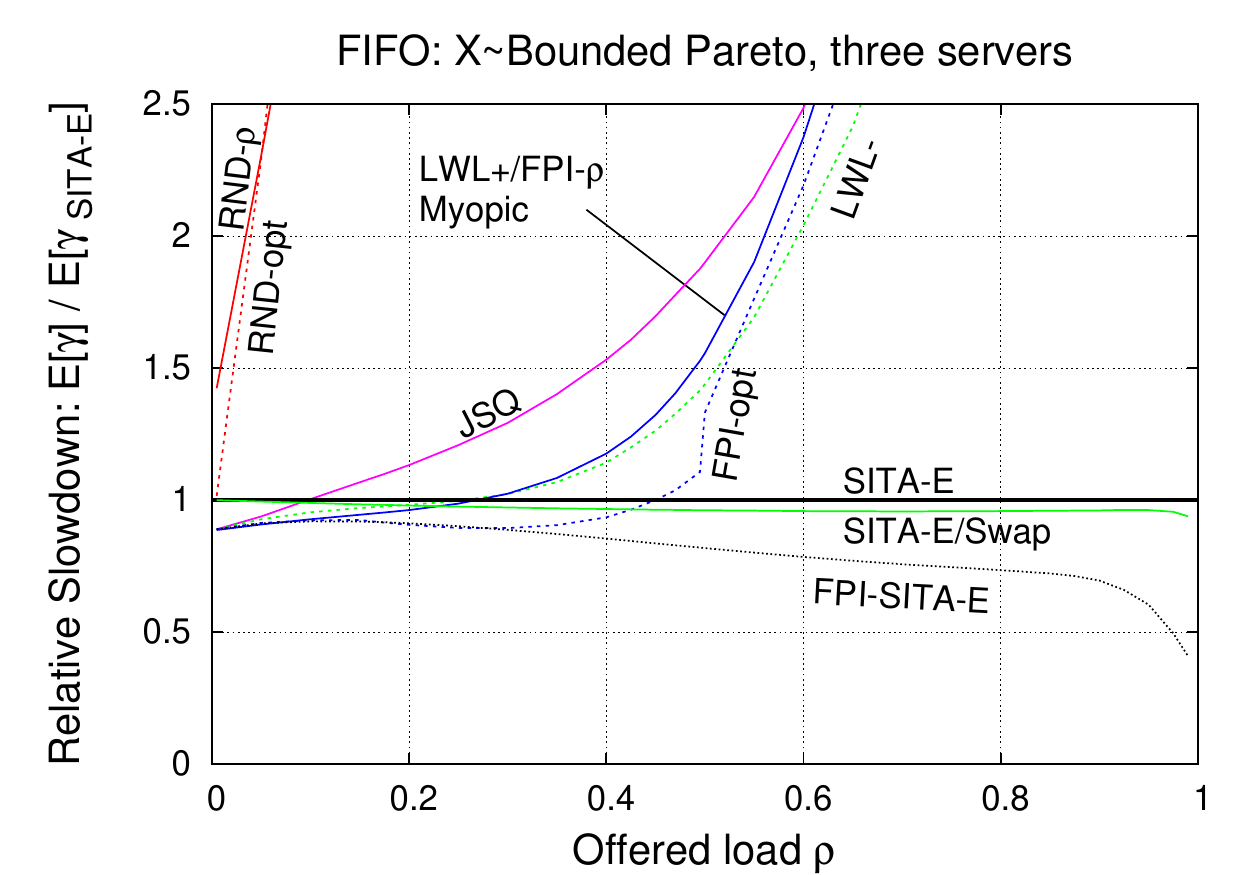}
\caption{Numerical results with two identical (left) and
three heterogeneous servers having rates $\mathbf{\boldsymbol{\nu}=(1,0.5,0.5)}$ (right).
The scheduling discipline is FIFO and jobs sizes obey bounded Pareto distribution.}
\label{fig:fifo-results}
\end{figure*}

\subsection{FIFO}

First we assume
that servers %
operate under the FIFO scheduling discipline
and job sizes obey the bounded Pareto distribution \eqref{eq:pareto}.
With FIFO, the departure time gets fixed at dispatching
and the Myopic policy %
corresponds to
selfish users choosing the queue that guarantees the shortest sojourn time,
i.e., LWL$^+$.
Similarly, LWL$^-$ is equivalent to an M/G/$m$ system with a single shared queue.

\subsubsection*{Two identical servers}

Fig.~\ref{fig:fifo-results} (left) depicts the results with two identical servers
with rates
$\nu_1=\nu_2=1$.
The $x$-axis corresponds to the offered load $\rho$ and
the $y$-axis to the relative mean slowdown, $\E{\gamma}/\E{\gamma_{\mathrm{SITA-E}}}$,
i.e., the comparison is against SITA-E.
We find that 
the policies appear to fall in one of the three groups with respect to slowdown:
SITA policies form the best group, other queue state-dependent policies the next,
and RND-U and Round-Robin have the worst performance.
Especially, FPI-SITA-E outperforms the other policies by
a significant margin, including the other SITA policies.
Note also that when $\rho$ is small, the sensible state-dependent policies utilize
idle servers better than, e.g., SITA-E.

Fig.~\ref{fig:fpi-sita-e}
illustrates SITA-E and FPI-SITA-E in the same setting 
at an offered load of $\rho=0.5$.
The $x$- and $y$-axes correspond to the backlog in Queue 1 and Queue 2, respectively, and
the $z$-axis to the maximum job size %
a given policy assigns to Queue 1.
With SITA-E, this threshold is constant, while FPI-SITA-E changes the threshold dynamically.
One observes that as a result of FPI, 
Queue 1 has become a ``high-priority'' queue where no jobs are accepted whenever 
Queue 1 has more unfinished work than Queue 2.
Similarly, when Queue 2 
has a long backlog,
the threshold for assigning a job to Queue 1 becomes higher with FPI %
than %
with SITA-E.

\subsubsection*{Three heterogeneous servers}

Consider next the asymmetric setting
comprising one primary server with service rate $\nu_1=1$ and two secondary servers
with rates $\nu_2=\nu_3=1/2$.
Fig.~\ref{fig:fifo-results} (right) illustrates the results from a numerical simulation.
\begin{figure}
\centering
\includegraphics[width=60mm]{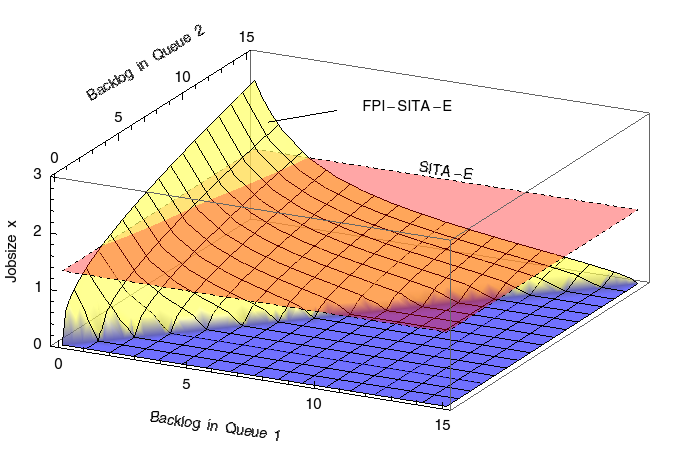}
\caption{SITA-E and FPI-SITA-E policies in the setting of
  two identical FIFO queues, bounded Pareto distributed job sizes and
  offered load $\mathbf{\boldsymbol{\rho}=0.5}$.
  $\mathbf{x}$- and $\mathbf{y}$-axes correspond to the backlog in Queue 1 and Queue 2, respectively, and
  $\mathbf{z}$-axis to the maximum job size that a policy assigns to Queue 1.}
\label{fig:fpi-sita-e}
\end{figure}

\begin{figure*}
\centering
\includegraphics[width=72mm]{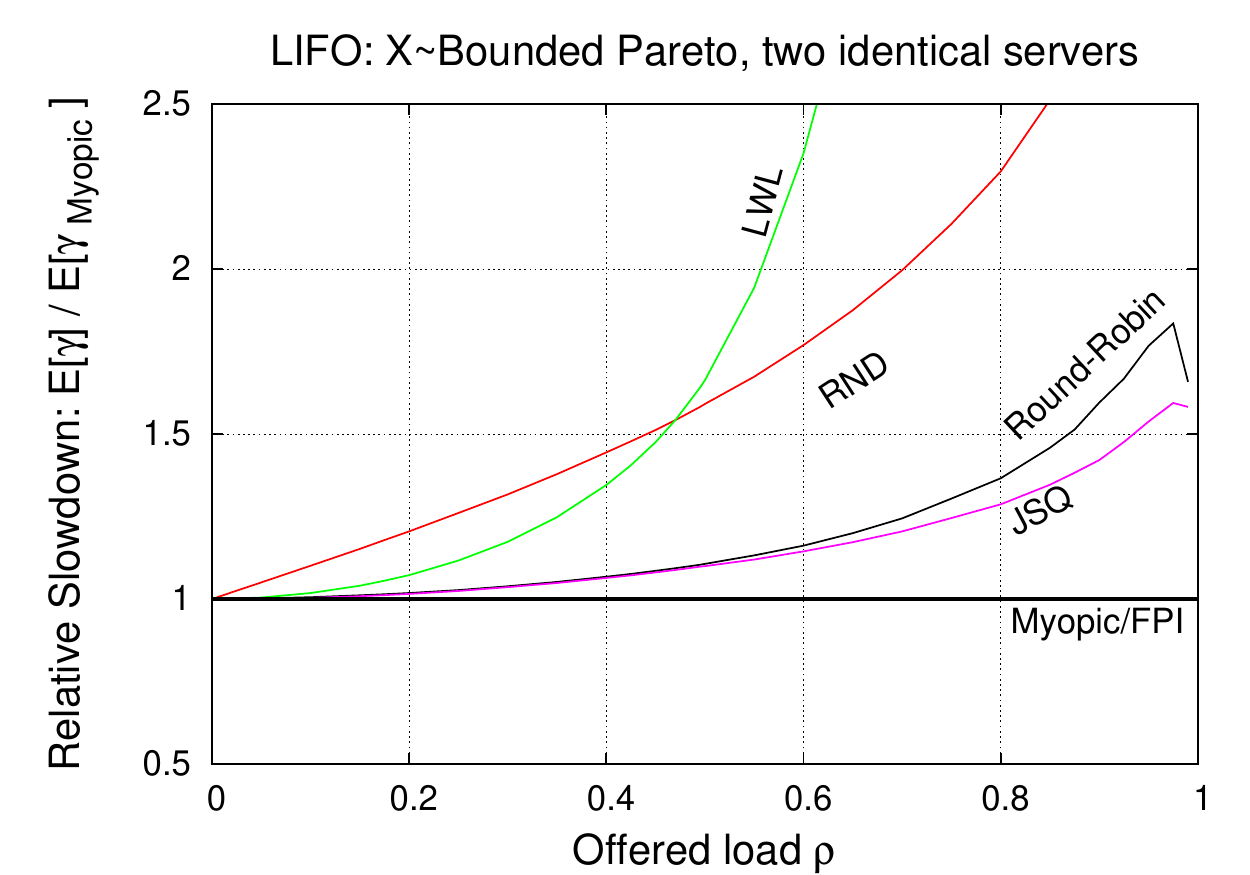}
\hspace{20mm}
\includegraphics[width=72mm]{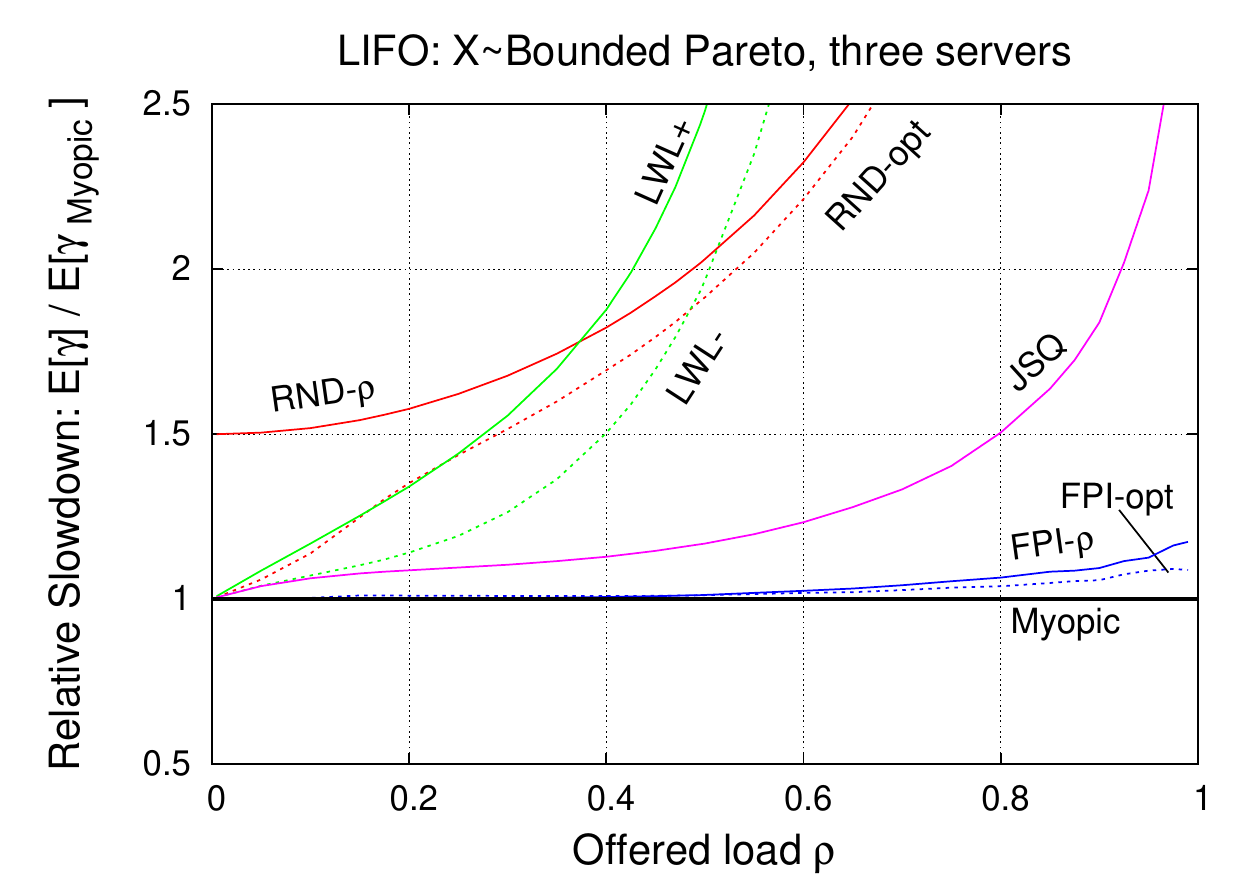}
\caption{Numerical results with preemptive LIFO scheduling discipline
  for two identical (left) servers with rates $\mathbf{\boldsymbol{\nu}_1=\boldsymbol{\nu}_2=1}$,
  and three heterogeneous servers with rates $\mathbf{\boldsymbol{\nu}_1=1}$ and 
  $\mathbf{\boldsymbol{\nu}_2=\boldsymbol{\nu}_3=1/2}$.}
\label{fig:plifo-sim}
\end{figure*}
The state-independent random policies are much worse than the other policies.
However, state-independent SITA-E again proves out to be better than the
state-dependent policies JSQ, LWL$^-$, LWL$^+$, Myopic, FPI-$\rho$ and FPI-opt.
The gain from switching the roles of the queues (SITA-E vs.\ SITA-Es)
is smaller in this case due to the fact that we may only switch
the roles of the two slower queues having the same service rate $\nu=0.5$.
As expected, the FPI-SITA-E policy yields a
significantly lower mean slowdown than any other policy
especially under a heavy load.

Based on the results, one can assume that the relative values
obtained for the random policies simply do not capture
the situation sufficiently well, while
SITA-E is already a good dispatching policy,
and FPI then leads to ``adaptive'' size intervals.

\subsection{Preemptive LIFO}

Next we assume that the servers are bound to operate under LIFO,
which, as mentioned, 
has a reasonably robust performance with respect to the mean slowdown.
Job sizes are again assumed to obey the bounded Pareto distribution.

\subsubsection*{Two identical servers}

Fig.~\ref{fig:plifo-sim} (left) illustrates the performance with respect to the slowdown
criterion for two identical servers. %
On $x$-axis is again the offered load $\rho$ and $y$-axis corresponds to the
relative mean slowdown (here comparison is against the Myopic policy).
One can identify three performance groups: RND and LWL form the worst performing group,
then come Round-Robin and JSQ, and
the Myopic and FPI-RND policy achieve the lowest mean slowdown.

\subsubsection*{Three heterogeneous servers}
Fig.~\ref{fig:plifo-sim} (right) illustrates the simulation results in the asymmetric setting
comprising one primary server with service rate $\nu_1=1$ and two secondary servers
with rates $\nu_2=\nu_3=1/2$. 
In this case, the Myopic approach is 
the optimal (among the candidates) while both FPI policies
attain almost identical mean slowdown. %
Even though the value function for the Myopic policy is not available for us,
one can estimate the relative values by means of simulation and carry out
the policy improvement numerically. %
However, in this paper we do not pursue into this direction.

\begin{figure}
\centering
\includegraphics[width=72mm]{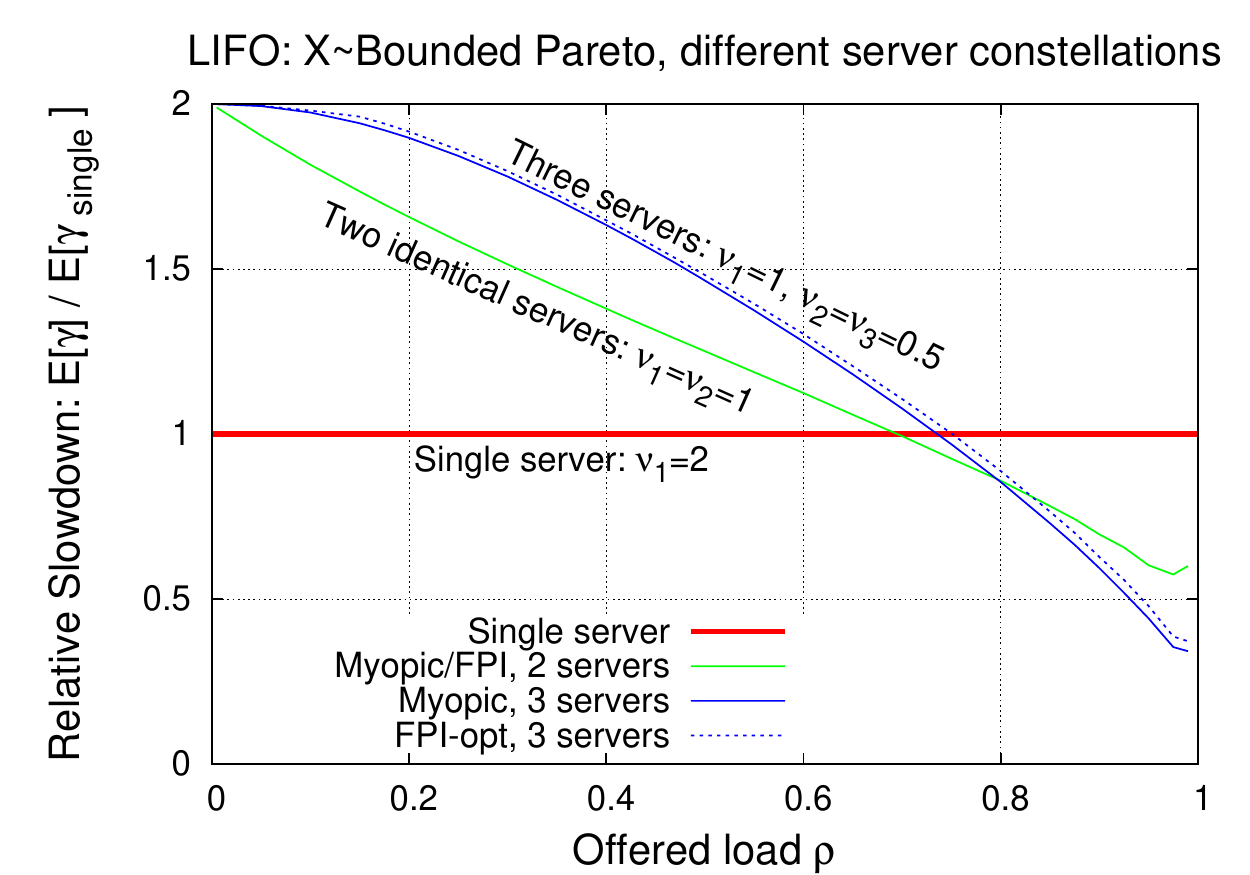}
\caption{Comparison of the chosen server constellations with an equal total capacity: 
  i) single server, ii) two identical servers, and iii) three servers.} %
\label{fig:plifo-comparison}
\end{figure}

\begin{figure*}
\centering
\begin{tabular}{c@{}c@{}c}
\includegraphics[width=67mm]{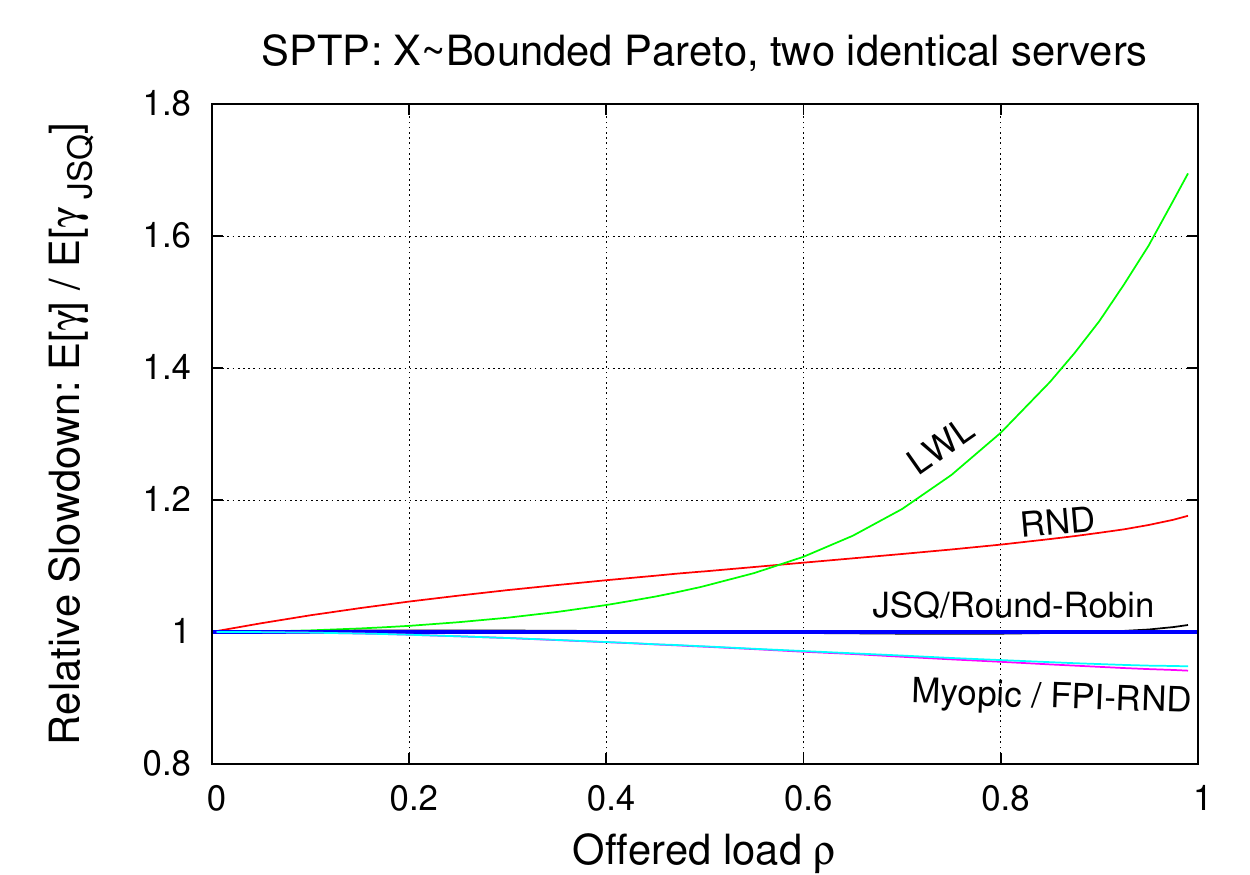}
& \hspace{26mm} &
\includegraphics[width=67mm]{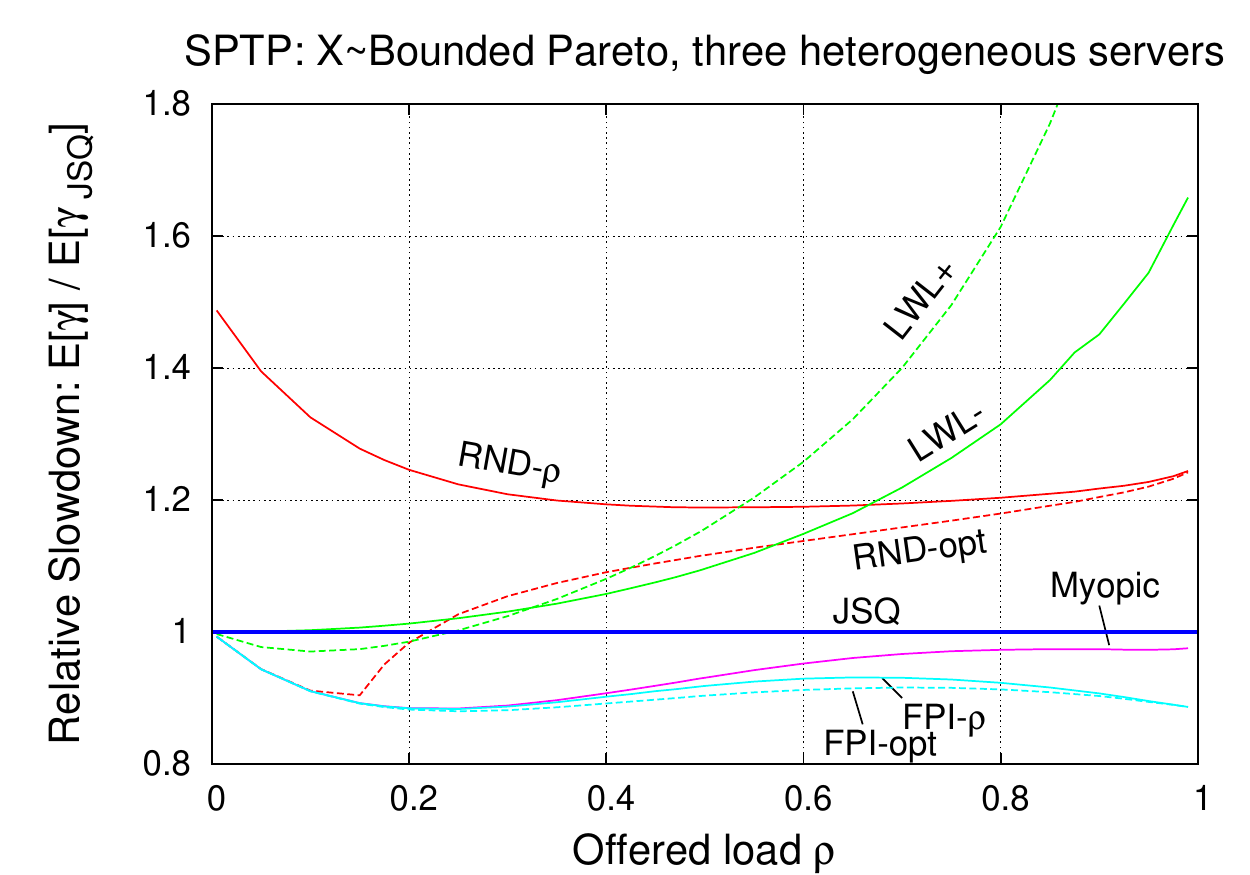}\\
\end{tabular}
\caption{Numerical results with SPTP for two identical (left)
  and three heterogeneous servers (right).}
\label{fig:sptp-sim-3}
\end{figure*}

\subsubsection*{Server constellations}
In Fig.~\ref{fig:plifo-comparison}
we compare the mean slowdown
between different server constellations
assuming the preemptive LIFO scheduling discipline.
The $y$-axis represents the relative performance when compared to the two server
system.
To no surprise, 
when the load is small or moderate, a single server system achieves the lowest mean slowdown.
However, as the load increases more, first the two server system becomes optimal,
and then eventually the three server system takes the lead.
This is due to the fact that at higher load levels, the correct dispatching decisions
allow one to serve more expensive jobs faster.
One interesting research question indeed is
that given a total capacity budget, what is the optimal
server constellation so as to minimize the mean slowdown.

\subsection{SPTP}

Finally, let us consider the SPTP scheduling discipline that is the natural
choice when minimizing the mean slowdown (see Section~\ref{sec:mg1-slow}). %
The job sizes are again assumed to obey the bounded Pareto distribution.
Fig.~\ref{fig:sptp-sim-3} (left) illustrates the slowdown
performance with the two identical servers.
Interestingly, LWL becomes even weaker than RND-U as the offered load increases.
Myopic and FPI-RND achieve the lowest mean slowdown.

\begin{figure*}
  \centering
  \includegraphics[width=67mm]{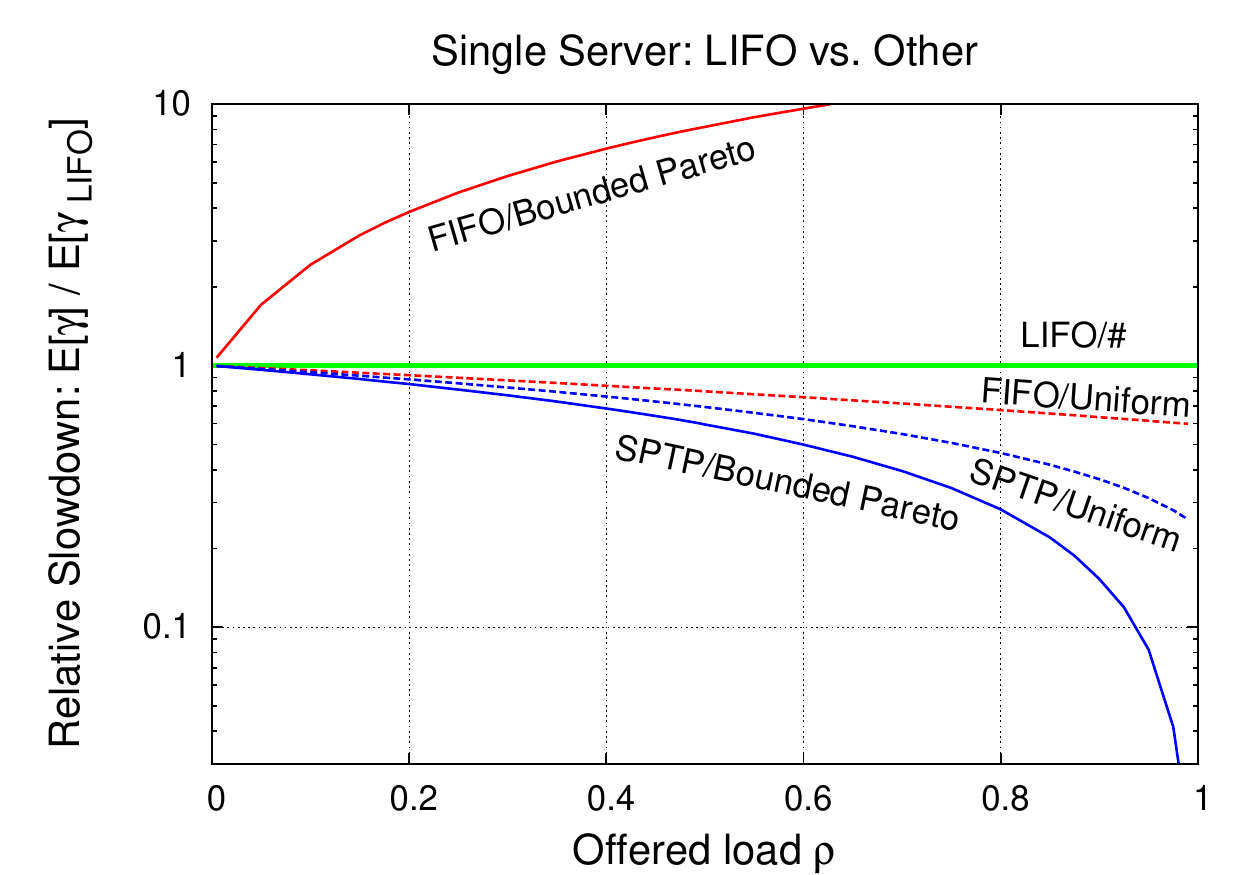}
  \hspace{26mm}
  \includegraphics[width=67mm]{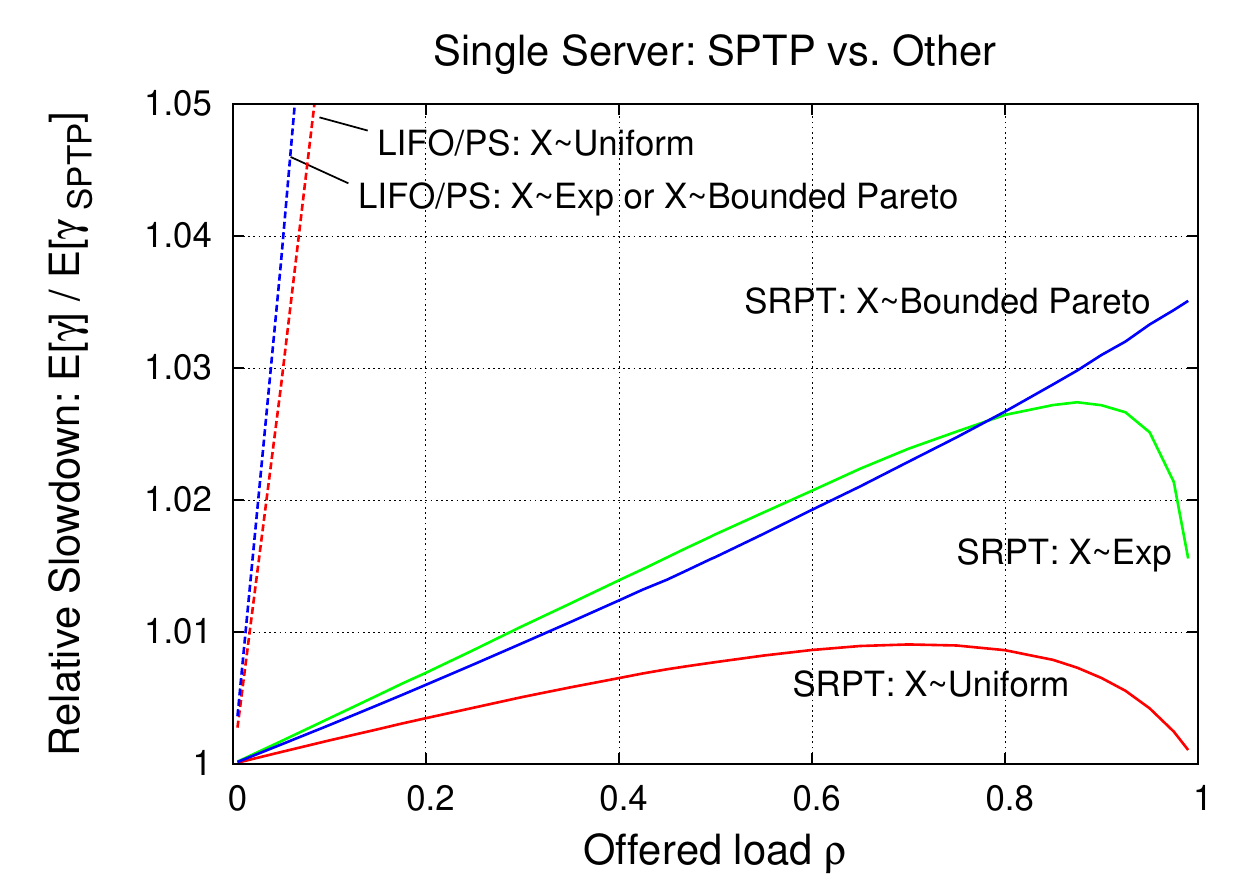}
  \caption{Comparison of different scheduling disciplines in a single server queue.
    FIFO is better than LIFO with some job size distributions (left).
    SPTP is only marginally better than SRPT (right).}
  \label{fig:comparisons}
\end{figure*}

In the case of three heterogeneous servers, the situation is again more challenging
and the numerical results are depicted in Fig.~\ref{fig:sptp-sim-3} (right).
At low levels of load, RND-opt is a surprisingly good choice,
suggesting that SPTP manages to locally ``correct'' the occasional suboptimal decisions. 
When the load increases, the LWL policies become weak again.
The Myopic policy shows a robust and good performance at all levels of offered load.
However, the FPI policies, and FPI-opt in particular, achieve the lowest mean slowdown which
is significantly better than with any other policy when $\rho > 0.5$.

\subsection{Comparison of scheduling disciplines}

\begin{figure*}
  \centering
  \begin{tabular}{c@{}c@{}c}
  \includegraphics[width=75mm]{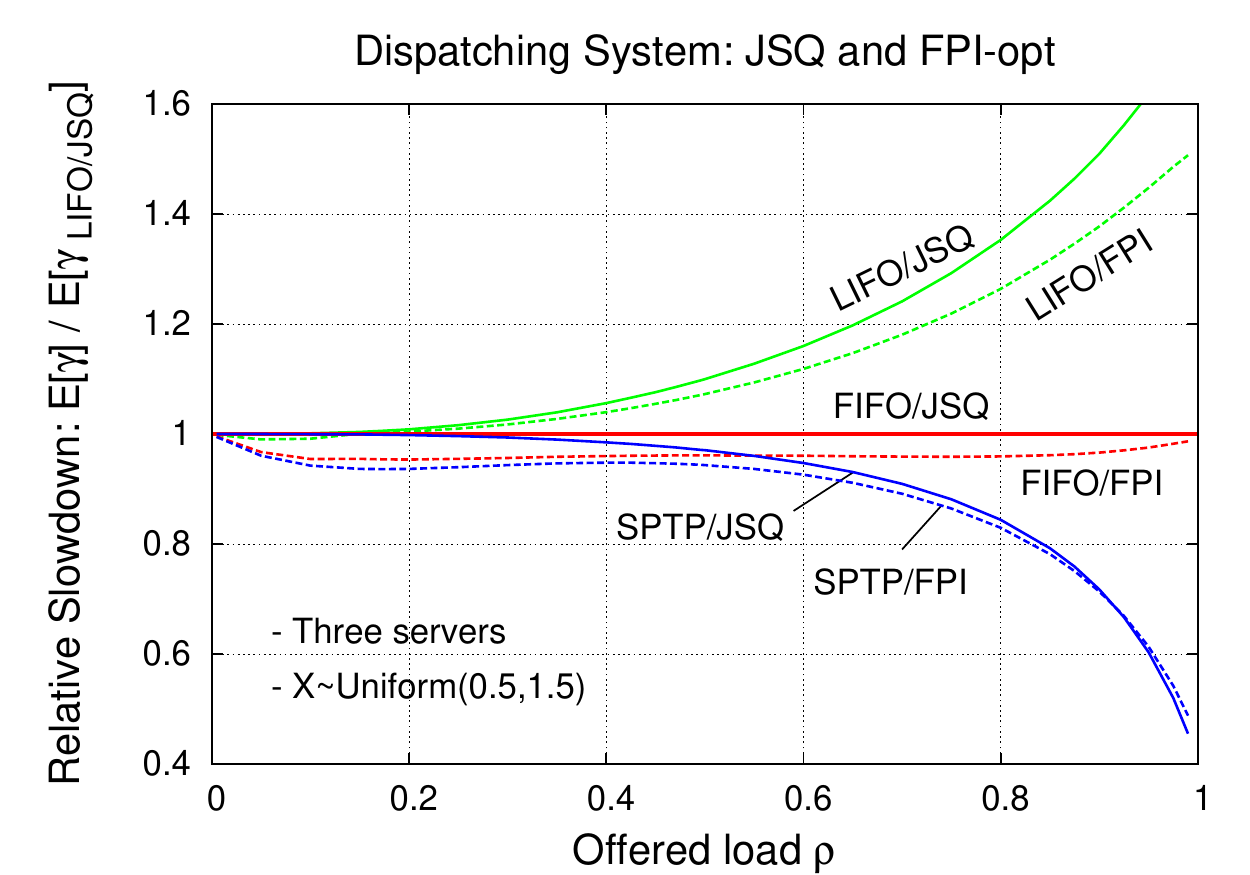} &
  \hspace{12mm} &
  \includegraphics[width=75mm]{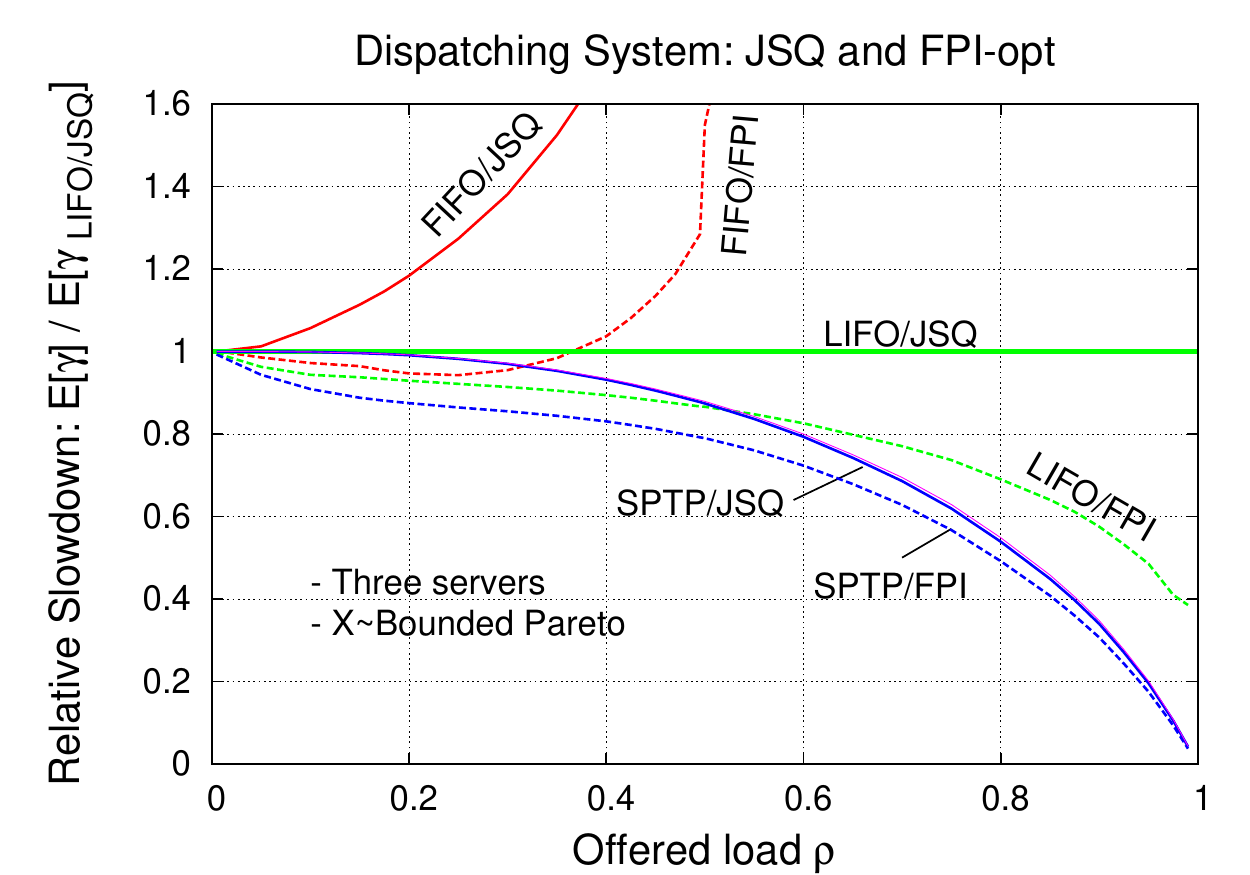} \\
  \small (a) $\JS \sim \mathrm{U}(0.5,1.5)$ &&
  \small (b) $\JS \sim \text{Bounded Pareto}$
  \end{tabular}
  \caption{Heterogeneous dispatching system with different scheduling disciplines.}
  \label{fig:dispatching-3-vertailu}
\end{figure*}

So far we have fixed a scheduling discipline and evaluated different dispatching policies.
Here we give a brief comparison of 
different scheduling disciplines with respect to the slowdown
metric. Job sizes obey either (a) uniform distribution, $Y\sim\mathrm{U}(0.5,1.5)$,
or (b) the bounded Pareto distribution, both having the same mean, $\E{Y}=1$.

First we consider a \emph{single server queue}. %
Fig.~\ref{fig:comparisons} (left) depicts a comparison
against LIFO with the same job size distribution on the logarithmic scale.
In accordance with \eqref{eq:fifo-vs-lifo}, 
FIFO is better %
than LIFO with uniform job size distribution,
and with the bounded Pareto distribution the situation is the opposite. 
SPTP is clearly the best in both cases. 
In Fig.~\ref{fig:comparisons} (right), the reference level is the mean slowdown with SPTP.
Here we have included also SRPT, which %
turns out to be only marginally better than SRPT,
as observed also in \cite{wierman-sigmetrics-2005}.
LIFO and PS are significantly worse.

Fig.~\ref{fig:dispatching-3-vertailu} illustrates the relative performance between 
different %
scheduling disciplines and 
dispatching policies in the heterogeneous three server dispatching system.
SRPT is not included as its performance is very similar to SPTP.
The reference level is the mean slowdown a JSQ/FIFO (left)
and JSQ/LIFO (right) %
achieve.
FPI manages to outperform the corresponding JSQ policy in all cases.
Especially with the bounded Pareto distributed job sizes, 
FPI/LIFO performs rather well when compared to JSQ/LIFO.
SPTP is clearly superior scheduling discipline in both cases, as expected.

In general, the scheduling discipline seems to have a stronger influence to the performance
than the dispatching policy.

\section{Conclusions}
\label{sec:conclusions}

This work generalizes the earlier results of the size-aware
value functions with respect to the sojourn time for M/G/1
queues with FIFO, LIFO, SPT and SRPT disciplines to a general
job specific holding cost. 
Additionally, we have considered the SPTP queueing discipline, 
the optimality of which with Poisson arrivals was also established.
As an important special case, one obtains the \emph{slowdown criterion},
where the holding cost is inversely proportional to the job's (original) size.
These results were then utilized in developing robust policies
for dispatching systems.
In particular, the value functions enable the policy improvement step for
an arbitrary state-independent dispatching policy such as Bernoulli-splitting.
The derived dispatching policies were also shown to perform well
by means of numerical simulations. In particular, the highest improvements
were often obtained in a more challenging heterogeneous setting with
unequal service rates. Also the SPTP scheduling discipline outperformed
the other by a clear margin in the examples, except the SRPT discipline,
which appears to offer a rather similar performance in terms of slowdown
and sojourn time.

\bibliographystyle{abbrv}

\begin{thebibliography}{10}

\bibitem{aalto-sigmetrics-2007}
S.~Aalto, U.~Ayesta, S.~Borst, V.~Misra, and R.~N{\'u}{\~n}ez-Queija.
\newblock Beyond processor sharing.
\newblock {\em {SIGMETRICS} Perform. Eval. Rev.}, 34:36--43, 2007.

\bibitem{aalto-queueing-2009}
S.~Aalto, U.~Ayesta, and R.~Righter.
\newblock On the {Gittins} index in the {M/G/1} queue.
\newblock {\em Queueing Systems}, 63(1-4):437--458, 2009.

\bibitem{aalto-nts13-1996}
S.~Aalto and J.~Virtamo.
\newblock Basic packet routing problem.
\newblock In {\em {NTS-13}}, Trondheim, Norway, Aug. 1996.

\bibitem{bachmat-peva-2009}
E.~Bachmat and H.~Sarfati.
\newblock Analysis of {SITA} policies.
\newblock {\em Performance Evaluation}, 67(2):102--120, 2010.

\bibitem{bell-management-1983}
C.~E. Bell and J.~Shaler~Stidham.
\newblock Individual versus social optimization in the allocation of customers
  to alternative servers.
\newblock {\em Management Science}, 29(7):831--839, July 1983.

\bibitem{crovella-sigmetrics-1998}
M.~E. Crovella, M.~Harchol-Balter, and C.~D. Murta.
\newblock Task assignment in a distributed system: Improving performance by
  unbalancing load.
\newblock In {\em {ACM} {SIGMETRICS}}, pages 268--269, June 1998.

\bibitem{downey-hpdc-1997}
A.~Downey.
\newblock A parallel workload model and its implications for processor
  allocation.
\newblock In {\em {IEEE} {HPDC}'97}, pages 112--123, Aug. 1997.

\bibitem{ephremides-tac-1980}
A.~Ephremides, P.~Varaiya, and J.~Walrand.
\newblock A simple dynamic routing problem.
\newblock {\em {IEEE} Transactions on Automatic Control}, 25(4):690--693, Aug.
  1980.

\bibitem{feng-sigmetrics-2003}
H.~Feng and V.~Misra.
\newblock Mixed scheduling disciplines for network flows.
\newblock {\em {SIGMETRICS} Perform. Eval. Rev.}, 31:36--39, Sept. 2003.

\bibitem{feng-peva-2005}
H.~Feng, V.~Misra, and D.~Rubenstein.
\newblock Optimal state-free, size-aware dispatching for heterogeneous
  {M/G/}-type systems.
\newblock {\em Perform. Eval.}, 62(1-4), 2005.

\bibitem{gittins-1989}
J.~Gittins.
\newblock {\em Multi-armed Bandit Allocation Indices}.
\newblock Wiley, 1989.

\bibitem{gupta-peva-2007}
V.~Gupta, M.~Harchol-Balter, K.~Sigman, and W.~Whitt.
\newblock Analysis of join-the-shortest-queue routing for web server farms.
\newblock {\em Performance Evaluation}, 64(9-12):1062--1081, Oct. 2007.

\bibitem{harchol-balter-jacm-2002}
M.~Harchol-Balter.
\newblock Task assignment with unknown duration.
\newblock {\em J. ACM}, 49(2):260--288, Mar. 2002.

\bibitem{harchol-balter-pdc-1999}
M.~Harchol-Balter, M.~E. Crovella, and C.~D. Murta.
\newblock On choosing a task assignment policy for a distributed server system.
\newblock {\em Journal of Parallel and Distributed Computing}, 59:204--228,
  1999.

\bibitem{harchol-balter-cs-1997}
M.~Harchol-Balter and A.~B. Downey.
\newblock Exploiting process lifetime distributions for dynamic load balancing.
\newblock {\em {ACM} Transactions on Computer Systems}, 15(3):253--285, Aug.
  1997.

\bibitem{harchol-balter-sigmetrics-2009}
M.~Harchol-Balter, A.~Scheller-Wolf, and A.~R. Young.
\newblock Surprising results on task assignment in server farms with
  high-variability workloads.
\newblock In {\em {ACM} {SIGMETRICS}}, pages 287--298, 2009.

\bibitem{harchol-balter-peva-2002}
M.~Harchol-Balter, K.~Sigman, and A.~Wierman.
\newblock Asymptotic convergence of scheduling policies with respect to
  slowdown.
\newblock {\em Perform. Eval.}, 49(1-4):241--256, Sept. 2002.

\bibitem{haviv-orl-2007}
M.~Haviv and T.~Roughgarden.
\newblock The price of anarchy in an exponential multi-server.
\newblock {\em Oper. Res. Lett.}, 35(4):421--426, 2007.

\bibitem{hyytia-ejor-2012}
E.~Hyyti{\"a}, A.~Penttinen, and S.~Aalto.
\newblock Size- and state-aware dispatching problem with queue-specific job
  sizes.
\newblock {\em European Journal of Operational Research}, 217(2):357--370, Mar.
  2012.

\bibitem{hyytia-itc-2011}
E.~Hyyti{\"a}, A.~Penttinen, S.~Aalto, and J.~Virtamo.
\newblock Dispatching problem with fixed size jobs and processor sharing
  discipline.
\newblock In {\em {ITC'23}}, SFO, USA, Sept. 2011.

\bibitem{hyytia-peva-2011}
E.~Hyyti{\"a}, J.~Virtamo, S.~Aalto, and A.~Penttinen.
\newblock {M/M/1-PS} queue and size-aware task assignment.
\newblock {\em Performance Evaluation}, 68(11):1136--1148, Nov. 2011.

\bibitem{krishnan-ieee-ac-1990}
K.~R. Krishnan.
\newblock Joining the right queue: a state-dependent decision rule.
\newblock {\em {IEEE} Transactions on Automatic Control}, 35(1):104--108, Jan.
  1990.

\bibitem{liu-or-1998}
Z.~Liu and R.~Righter.
\newblock Optimal load balancing on distributed homogeneous unreliable
  processors.
\newblock {\em Operations Research}, 46(4):563--573, 1998.

\bibitem{liu-applied-1994}
Z.~Liu and D.~Towsley.
\newblock Optimality of the round-robin routing policy.
\newblock {\em Journal of Applied Probability}, 31(2):466--475, June 1994.

\bibitem{schrage-or-1968}
L.~Schrage.
\newblock A proof of the optimality of the shortest remaining processing time
  discipline.
\newblock {\em Operations Research}, 16(3), 1968.

\bibitem{whitt-or-1986}
W.~Whitt.
\newblock Deciding which queue to join: Some counterexamples.
\newblock {\em Oper. Res.}, 34(1):55--62, 1986.

\bibitem{wierman-or-survey-2011}
A.~Wierman.
\newblock Fairness and scheduling in single server queues.
\newblock {\em Surveys in Operations Research and Management Science},
  16(1):39--48, 2011.

\bibitem{wierman-sigmetrics-2005}
A.~Wierman, M.~Harchol-Balter, and T.~Osogami.
\newblock Nearly insensitive bounds on {SMART} scheduling.
\newblock In {\em {ACM} {SIGMETRICS}}, pages 205--216, 2005.

\bibitem{winston-applied-1977}
W.~Winston.
\newblock Optimality of the shortest line discipline.
\newblock {\em Journal of Applied Probability}, 14:181--189, 1977.

\bibitem{yang-infocom-2002}
S.~Yang and G.~de~Veciana.
\newblock Size-based adaptive bandwidth allocation: optimizing the average
  {QoS} for elastic flows.
\newblock In {\em {IEEE} {INFOCOM}}, 2002.

\end{thebibliography}

%

\appendix

\section{Value Function for SPT and SRPT}
Carrying out the similar steps as with SPTP, 
it is straightforward to derive value function
also in the case of SPT and SRPT policies.
Again, we let
$f(x)$ denote the pdf of the job size distribution and
$\rho(x)$ the offered load due to jobs shorter than $x$
according to \eqref{eq:rho-x}.
The backlog due to higher priority work in this case is
$u_{\bz}(x)$, i.e., where the priority index for SPT is the initial
service time $\Delta_i^*$,
and for SRPT the remaining service time $\Delta_i$.

\subsection{M/G/1-SPT}
We consider the non-preemptive SPT, which is the optimal
non-preemptive schedule with respect to both the mean sojourn time
and the mean slowdown.
Thus, the remaining service time of a queueing
job is equal to the initial service time.
Therefore, a sufficient state description for
a non-preemptive M/G/1-SPT queue with
arbitrary holding costs is
$\bz=((\Delta_1,a_1);\ldots;(\Delta_n,a_n))$,
where job $1$ (if any) is currently receiving service
and jobs $2,\ldots,n$ are waiting in the queue.
Without lack of generality, we can assume that
$\Delta_2<\Delta_3<\ldots<\Delta_n$.

\begin{prop}%
For the size-aware value function
with respect to arbitrary holding costs
in an non-preemptive M/G/1-SPT queue it holds that,
\begin{equation}\label{eq:vn-spt}\small
\begin{array}{l}
v_{\bz} %
-v_0 = 
\dps \sum_{i=1}^n \hc_i\left(\Delta_i + \frac{\sum_{j=1}^{i-1} \Delta_j}{1-\rho(\Delta_i)} \right) \; +
\\[6mm]\dps
 \frac{\lambda}{2} \sum_{i=1}^{n}
\left( \left(\biggl(\sum_{j=1}^i \Delta_j\biggr)^2 {+} \sum_{j=i+1}^{n} (\Delta_j)^2\right)
\hspace*{-2mm}\int\limits_{\tilde\Delta_{i}}^{\tilde\Delta_{i+1}}\hspace*{-3mm}
  \frac{\hc(x)\,f(x)}{(1{-}\rho(x))^2}\,dx \right),
\end{array}
\end{equation}
where $\hc(x)=\cE{\HC}{\JS{=}x}$ and %
the integration intervals are
$$
\tilde\Delta_i = \left\{
\begin{array}{ll}
0, & i=1\\
\Delta_i, & i=2,\ldots,n\\
\infty, & i=n+1.
\end{array}\right.
$$
\end{prop}

\begin{proof}
Similarly as with the SPTP,
we compare two systems:
System 1 initially in state $\bz$ and System 2 initially empty.
The two systems behave identically once System 1 becomes empty
for the first time.
We can write $v_{\bz}-v_0 = h_1+h_2$, where
$h_1$ is the cost the $n$ jobs initially present (only) in System 1
incur, and $h_2$ the difference in cost the later arriving
customers incur between the two systems.

First, the remaining sojourn time of job $i$ in a non-preemptive M/G/1-SPT queue
is
\begin{equation}\label{eq:spt-virtual}
\E{R_i} = \Delta_i + \frac{\sum_{j=1}^{i-1} \Delta_j}{1-\rho(\Delta_i)},
\end{equation}
which holds for all $n$ jobs present only in System 1.
Therefore, their contribution to the $v_{\bz}-v_0$ is
$$
h_1 = \sum_{i=1}^n \hc_i \left( \Delta_i + \frac{\sum_{j=1}^{i-1} \Delta_j}{1-\rho(\Delta_i)} \right).
$$

For the later arriving jobs,
we condition the derivation on job sizes $(x,x+dx)$,
which arrive at rate of $\lambda\,f(x)\,dx$.
Let $\tilde{t}(x)$ denote the mean additional sojourn time such
jobs experience in System 1 when compared to System 2.
Let $\hc(x)=\cE{\HC}{\JS{=}x}$ so that we have
an expression for $h_2$
$$
h_2 = \int_0^\infty \tilde{t}(x)\cdot \hc(x)\,dx.
$$
Instead of considering actual arrivals, we focus on a rate at
which \emph{virtual costs} are accrued in order to find $\tilde{t}(x)$.
With aid of \eqref{eq:spt-virtual}, one can deduce that 
the virtual cost rate is equal to
$$
\lambda f(x) \, \frac{U_{\bz}(x,t)}{1-\rho(x)},
$$
where $U_{\bz}(x,t)$ denotes the amount of work at time $t$
that would be processed before an arriving size $x$ job 
when a system was initially in state $\bz$.
In particular, we can write
$$
\tilde{t}(x) = \E{ \int_0^\infty \lambda\,f(x)
  \left( \frac{U_{\bz}(x,t)}{1-\rho(x)} - \frac{U_{0}(x,t)}{1-\rho(x)} \right)\,dt },
$$
which gives
$$
\tilde{t}(x) = \frac{\lambda\,f(x)}{1-\rho(x)}\, \E{ \int_0^\infty U_{\bz}(x,t) - U_{\bz}(x,t) \,dt }.
$$

Let $k_{\bz}(x)$ denote the number of jobs waiting in the queue with
a remaining service time greater than $x$.
$$
k_{\bz}(x) = | \{ i\in\{2,\ldots,n\} \,:\,\Delta_i>x \} |,
$$
Due to the non-preemptive discipline, the evolution of 
$U_{\bz}(x,t)$ for $\bz\ne 0$ consists of $k_{\bz}(x)+1$ phases.
Let $m=n-k_{\bz}(x)$.
During the first phase jobs $1,\ldots,m$ are served
and the initial backlog size $x$ jobs see is $\Delta_1+\ldots+\Delta_m$.
The second phase starts when job $m+1$ enters the server
and cannot be preempted,
thus creating an initial backlog $\Delta_{m+1}$ for size $x$ jobs,
etc.

\begin{figure}
\includegraphics[width=\linewidth]{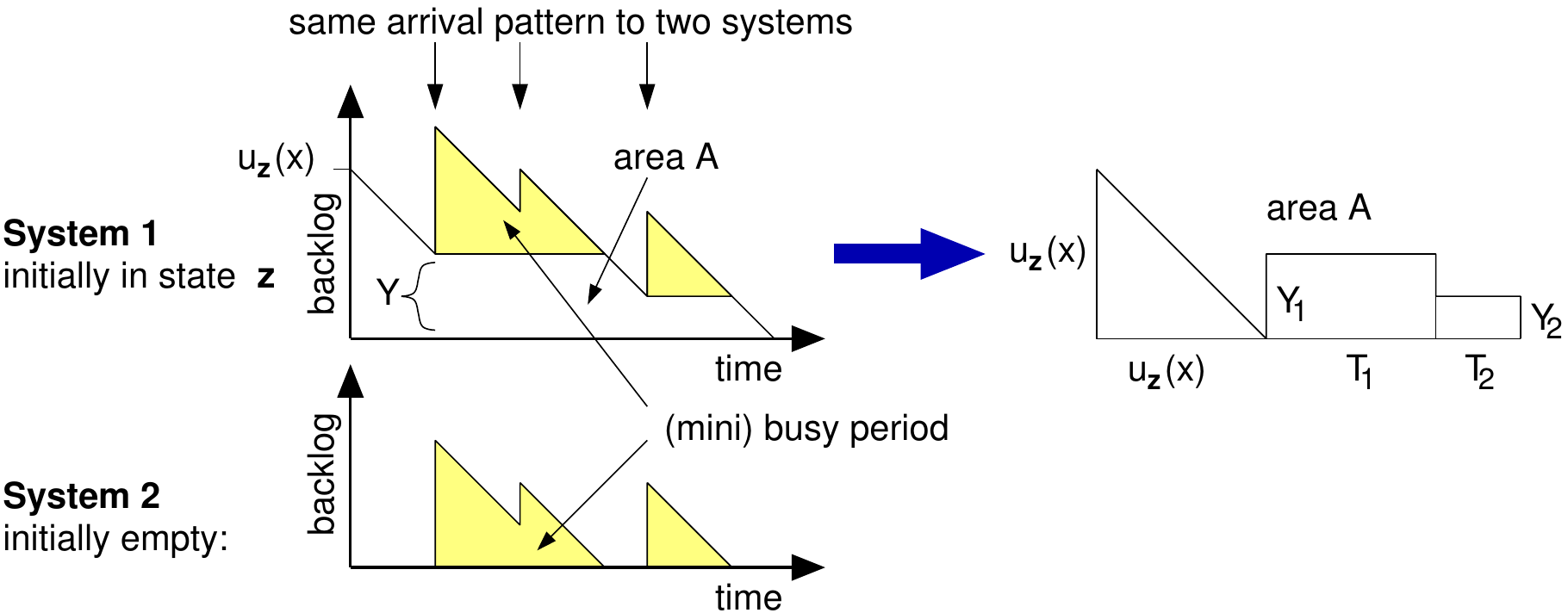}
\caption{Derivation of value function for an M/G/1-SPT queue.}
\label{fig:vn-spt}
\end{figure}

Considering arrival sample paths, one can deduce that the difference
$\int_0^\infty U_{\bz}(x,t) - U_{\bz}(x,t)\,dt$ corresponds to the white
marked region in Fig.~\ref{fig:vn-spt}, which consists of $k_{\bz}(x)+1$ statistically independent
areas.
The area of a phase starting with an initial backlog of $u$ from 
size $x$ customers' point of view is given by
\begin{equation}\label{eq:spt-vn-apu}
\frac{u^2}{2} + \lambda\,F(x)\,u \cdot \frac{\cE{\JS}{\JS<x}}{1-\rho}\cdot \frac{u}{2}
 = \frac{u^2}{2(1-\rho(x))}.
\end{equation}
Defining
$$
g(x)=\frac{\hc(x)\,f(x)}{(1-\rho(x))^2},
$$
then gives
\begin{align*}
h_2 &= \frac{\lambda}{2} ( \Delta_1^2 + \ldots + \Delta_n^2 ) \int_0^{\Delta_{2}} g(x)\,dx\\
 &\,+\,
\frac{\lambda}{2} ( (\Delta_1 + \Delta_2)^2 + \Delta_3^2 + \ldots + \Delta_n^2 ) \int_{\Delta_{2}}^{\Delta_{3}} g(x)\,dx + \ldots \\
 &\,+\,
\frac{\lambda}{2} ( \Delta_1 + \ldots + \Delta_n)^2 \int_{\Delta_{n}}^{\infty} g(x)\,dx.
\end{align*}
For example, when $x>\Delta_n$, all current $n$ jobs in System 1 have a higher
priority constituting a single phase with $u=\Delta_1+\ldots+\Delta_n$, etc.
Next define the integration intervals,
$\tilde\Delta_1=0$,
$\tilde\Delta_i=\Delta_i$ for $i=2,\ldots,n$, and
$\tilde\Delta_{n+1}=\infty$.
Substituting these into the above 
gives
\begin{equation*}%
\small
h_2 =
 \frac{\lambda}{2} \sum_{i=1}^{n}
 \left( \bigl(\sum_{j=1}^{i} \Delta_j\bigr)^2 {+} \sum_{j=i+1}^n \Delta_j^2\bigr)^2\right)
   \int\limits_{\tilde\Delta_{i}}^{\tilde\Delta_{i+1}}\hspace*{-3mm} \frac{\hc(x)\,f(x)}{(1-\rho(x))^2}\,dx,
\end{equation*}
which completes the proof.
\end{proof}

\subsection{M/G/1-SRPT}
While SPT was optimal in the class of non-preemptive schedules,
the SRPT discipline is the optimal preemptive schedule \cite{schrage-or-1968}.
For SRPT, a sufficient
state description is $\bz=((\Delta_1,a_1);\ldots;(\Delta_n,a_n))$,
where, without loss of generality, we again can assume that job $1$
is currently receiving service (if any)
and $\Delta_1<\Delta_2<\ldots<\Delta_n$.
For the total higher priority workload
we have $u_{\bz}(\Delta_i) = \sum_{j=1}^{i-1} \Delta_j$.

\begin{prop}%
\label{prop:vn-mg1-srpt}
For the size-aware value function 
in an M/G/1-SRPT queue with arbitrary holding cost
it holds that
\begin{equation}\label{eq:vn-srpt}\small
\begin{array}{l@{\,}l}
v_{\bz} - v_0 &= 
\dps 
 \sum_{i=1}^{n} \hc_i \left( \frac{\sum_{j=1}^{i-1} \Delta_j}{1{-}\rho(\Delta_i)} +
\int_0^{\Delta_i} \frac{1}{1{-}\rho(t)} \,dt\right) 
\\[6mm]&\dps
 + \;
\frac{\lambda}{2}\sum_{i=0}^n \biggl[ (n-i)
\int_{\tilde\Delta_i}^{\tilde\Delta_{i+1}} \frac{ x^2\,\hc(x)\,f(x)}{(1-\rho(x))^2}\,dx
\\[6mm]&\qquad\dps
 +\; \left(\sum_{j=1}^{i-1}\Delta_j \right)
 \int_{\tilde\Delta_i}^{\tilde\Delta_{i+1}} \frac{ \hc(x)\,f(x)}{(1-\rho(x))^2}\,dx
\biggr]
\end{array}
\end{equation}
where $\hc(x)=\cE{\HC}{\JS=x}$ and %
$$
\tilde\Delta_i = \left\{
\begin{array}{ll}
0, & i=0 \\
\Delta_i, & i=1,\ldots,n\\
\infty, & i=n+1.
\end{array}\right.
$$
\end{prop}

\begin{proof}
Again, 
we consider System 1 initially in state $\bz$ and System 2 initially empty,
and find expressions for $h_1$ and $h_2$
corresponding to
the cost the $n$ jobs initially present (only) in System 1
incur, and the difference in cost the later arriving
customers incur between the two systems. %

In a M/G/1-SRPT queue at state ${\bz}=(\Delta_1,\ldots,\Delta_n)$,
the mean remaining sojourn time of a job with a (remaining) size $\Delta$ and
an unfinished work $u_{\bz}(\Delta)$ ahead in the queue 
is given by \cite{hyytia-ejor-2012}
\begin{equation}\label{eq:et-srpt}
\E{R_{\bz}(\Delta)} = \frac{u_{\bz}(\Delta)}{1-\rho(\Delta)} +
\int_0^\Delta \frac{1}{1-\rho(t)}\, dt,
\end{equation}
which gives the expected cost
the current $n$ jobs\footnote{Note that we assume the opposite order than \cite{hyytia-ejor-2012}.}
present in System 1 incur,
$$
h_1 = \sum_{i=1}^{n} \hc_i \left(\frac{\sum_{j=1}^{i-1}\Delta_j}{1-\rho(\Delta_i)} +
\int_0^{\Delta_i} \frac{1}{1-\rho(t)}\, dt\right).
$$
For the later arriving jobs, %
we can again write
$$
\begin{array}{l@{\,}l}
\tilde{t}(x) &= 
\dps
\mathrm{E}\biggl[ \int_0^\infty \lambda\,f(x) \biggl[
\bigl(x + \frac{U_{\bz}(x,t)}{1-\rho(x)} + \int_0^x \frac{\rho(t)}{1-\rho(t)}\,dt\bigr)
\\
&\dps\qquad\qquad -\;
\left(x + \frac{U_{0}(x,t)}{1-\rho(x)} + \int_0^x \frac{\rho(t)}{1-\rho(t)}\,dt\right)
\biggr] \biggr],\\[4mm]
 &=\dps
 \frac{\lambda\,f(x)}{1-\rho(x)}\, \mathrm{E}\biggl[\int_0^\infty U_{\bz}(x,t) - U_{0}(x,t)\,dt\biggr],
\end{array}
$$
which describes the expected difference in the cumulative sojourn times between
System 1 and System 2 for size $x$ jobs.

\begin{figure}
\includegraphics[width=\linewidth]{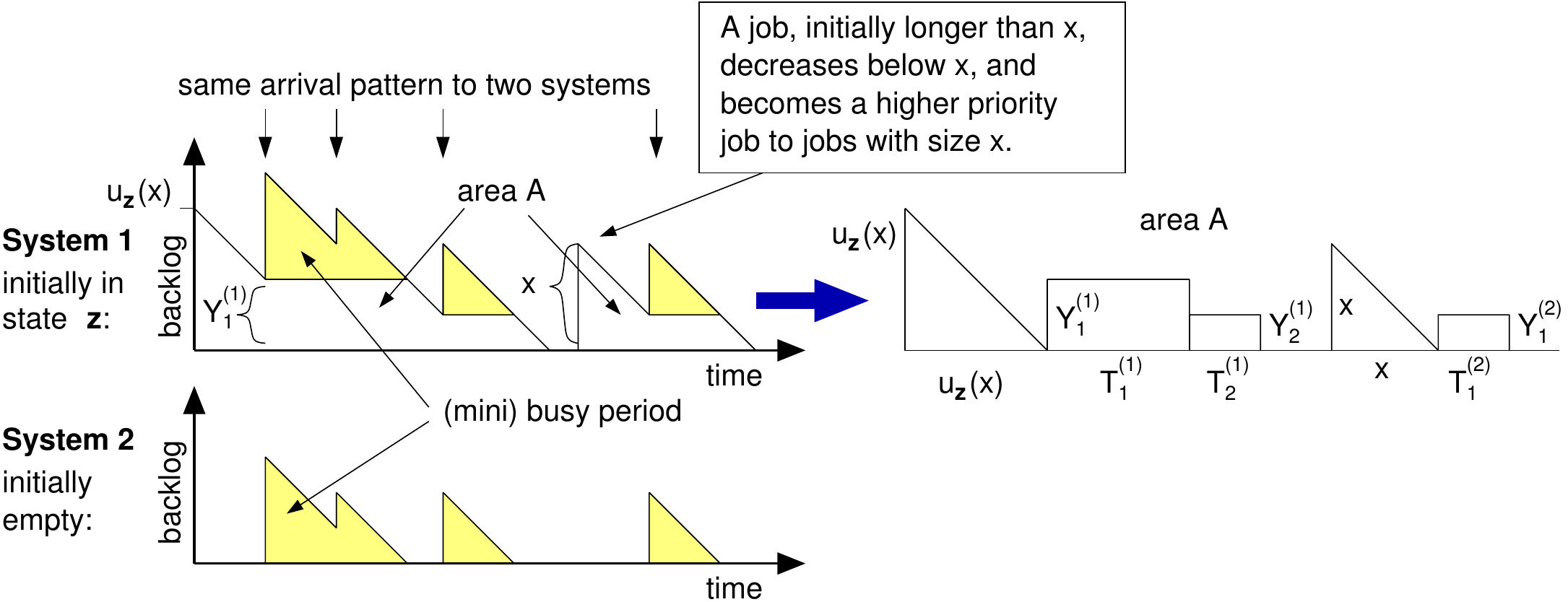}
\caption{Derivation of value function for an M/G/1-SRPT queue.}
\label{fig:vn-srpt}
\end{figure}

With SRPT also the jobs in System 1 initially longer than $x$
affect the result as eventually, one at a time, 
they decrease below the threshold $x$ and
trigger a new higher priority workload present only in System 1
(see Fig.~\ref{fig:vn-srpt}).
Consequently, 
we have one phase starting with a backlog of $u_{\bz}(x)$, and
then $n_{\bz}(x)$ phases starting with a backlog of $x$,
where $n_{\bz}(x)$ denotes the number of jobs longer than $x$ in state $\bz$.
In some sense, this is the same as with SPT, but the difference is the initial
backlog size $x$ jobs see: with SPT it is $\Delta_i$ instead of $x$ for later phases.
Hence, each of these phases reduce to \eqref{eq:spt-vn-apu}, which gives
$$
\tilde{t}(x) = \frac{\lambda\,f(x)\,\left[u_{\bz}(x)^2 + n_{\bz}(x)\,x^2\right]}{2 (1-\rho(x))^2},
$$
As $h_2 = \int \hc(x)\,\tilde{t}(x)\,dx$,
and both $u_{\bz}(x)$ and $n_{\bz}(x)$
are (monotonic) step functions having
jumps at $x=\Delta_1,\ldots,\Delta_n$,
we proceed by integrating in parts.
Defining
$$
g(x) = \frac{f(x)\,\hc(x)}{(1-\rho(x))^2},
$$
then gives
\begin{align*}
h_2 &= \frac{\lambda}{2} \int_0^{\Delta_1} n\,x^2 \, g(x)\,dx \\
& \quad +\; \frac{\lambda}{2} \int_{\Delta_1}^{\Delta_2} (\Delta_1^2 + (n-1)\,x^2) \, g(x)\,dx + \ldots \\
& \quad +\;
\frac{\lambda}{2} \int_{\Delta_n}^{\infty} (\Delta_1+\ldots+\Delta_n)^2 \, g(x)\,dx,
\end{align*}
and recalling that $v_{\bz}-v_0=h_1+h_2$ then gives \eqref{eq:vn-srpt}.
\end{proof}

\end{document}